\documentclass{article}

\usepackage[T1]{fontenc}
\usepackage[utf8]{inputenc}

\usepackage[margin=1in]{geometry}
\usepackage{graphicx}
\usepackage{subcaption}
\usepackage{authblk}

\usepackage{amsmath,amssymb,amsthm}

\usepackage[
natbib,
maxcitenames=3, 
maxbibnames=10,  
sorting=nyt,
url=false,
sortcites,
defernumbers,
backref,
backend=biber
]{biblatex}
\usepackage{url}
\usepackage[pdftitle={Proximal Langevin Sampling With Inexact Proximal Mapping},
pdfauthor={M. J. Ehrhardt, L. Kuger, C.-B. Schönlieb}
]{hyperref}

\usepackage[nameinlink, capitalise, noabbrev]{cleveref}
\crefname{assumption}{Assumption}{Assumptions}
\crefname{proposition}{Proposition}{Propositions}
\crefname{lemma}{Lemma}{Lemmata}
\crefname{theorem}{Theorem}{Theorems}
\crefname{remark}{Remark}{Remarks}
\crefname{algocf}{Algorithm}{Algorithms}

\usepackage{enumitem}

\usepackage{algorithm}
\usepackage{algorithmic} 
\Crefname{ALC@unique}{Line}{Lines}

\makeatletter
\let\blx@rerun@biber\relax
\makeatother

\DeclareMathOperator{\prox}{prox}
\DeclareMathOperator{\interior}{int}

\DeclareMathOperator{\sign}{sign}

\DeclareMathOperator*{\argmin}{arg\,min}
\DeclareMathOperator*{\argmax}{arg\,max}

\DeclareMathOperator{\TV}{TV}
\DeclareMathOperator{\domain}{dom}

\DeclareMathOperator{\KL}{KL}
\renewcommand{\epsilon}{\varepsilon}

\newcommand{\norm}[1]{\lVert #1  \rVert}

\newcommand{\inpro}[2]{\langle #1,#2 \rangle}

\newcommand{\bbN}{\mathbb{N}}
\newcommand{\bbR}{\mathbb{R}}
\newcommand{\bbE}{\mathbb{E}}

\newcommand{\calH}{\mathcal{H}}
\newcommand{\calP}{\mathcal{P}}
\newcommand{\calX}{\mathcal{X}}
\newcommand{\calY}{\mathcal{Y}}
\newcommand{\calG}{\mathcal{G}}
\newcommand{\calE}{\mathcal{E}}

\newcommand{\calW}{\mathcal{W}}
\newcommand{\calL}{\mathcal{L}}
\newcommand{\calD}{\mathcal{D}}

\newcommand{\rmderiv}[2]{\frac{\mathrm{d}#1}{\mathrm{d}#2}}
\newcommand{\rmd}{\,\mathrm{d}}
\newcommand{\rmN}{\,\mathrm{N}}

\newtheorem{theorem}{Theorem}[section]

\newtheorem{lemma}[theorem]{Lemma}

\newtheorem{remark}[theorem]{Remark}
\newtheorem{example}[theorem]{Example}
\newtheorem{definition}[theorem]{Definition}
\newtheorem{assumption}[theorem]{Assumption}

\numberwithin{equation}{section}

\title{Proximal Langevin Sampling With Inexact Proximal Mapping}
\author{Matthias J. Ehrhardt}
\author[2,3,*]{Lorenz Kuger}
\author[4]{Carola-Bibiane Sch\"onlieb}
\affil[1]{\small Department of Mathematical Sciences, University of Bath, Claverton Down, Bath BA2 7AY, United Kingdom.}
\affil[2]{\small Department Mathematik, Friedrich-Alexander-Universität Erlangen-Nürnberg, Cauerstr. 11, 91058 Erlangen, Germany}
\affil[3]{\small Helmholtz Imaging, Deutsches Elektronen-Synchrotron DESY, Notkestr. 85, 22607 Hamburg, Germany}
\affil[4]{\small Department of Applied Mathematics and Theoretical Physics, University of Cambridge, Wilberforce Road, Cambridge CB3 0WA, United Kingdom.}
\affil[*]{\small \href{mailto:lorenz.kuger@desy.de}{lorenz.kuger@desy.de}}

\begin{document}

\maketitle

\begin{abstract}
In order to solve tasks like uncertainty quantification or hypothesis tests in Bayesian imaging inverse problems, we often have to draw samples from the arising posterior distribution.
For the usually log-concave but high-dimensional posteriors, Markov chain Monte Carlo methods based on time discretizations of Langevin diffusion are a popular tool. If the potential defining the distribution is non-smooth, these discretizations are usually of an implicit form leading to Langevin sampling algorithms that require the evaluation of proximal operators. For some of the potentials relevant in imaging problems this is only possible approximately using an iterative scheme. We investigate the behaviour of a proximal Langevin algorithm under the presence of errors in the evaluation of proximal mappings. We generalize existing non-asymptotic and asymptotic convergence results of the exact algorithm to our inexact setting and quantify the bias between the target and the algorithm's stationary distribution due to the errors. We show that the additional bias stays bounded for bounded errors and converges to zero for decaying errors in a strongly convex setting. We apply the inexact algorithm to sample numerically from the posterior of typical imaging inverse problems in which we can only approximate the proximal operator by an iterative scheme and validate our theoretical convergence results.
\end{abstract}

\textbf{Keywords.} Langevin Sampling, Bayesian Computation, Proximal Mapping, Markov Chain Monte Carlo

\textbf{MSC codes.} 65C05, 65C40, 65C60, 65J22, 68U10

\section{Introduction}
In imaging sciences, the problem of estimating an image from acquired data is often ill-posed or ill-conditioned, resulting in uncertainty about the true solution. The formulation of such inverse problems in a Bayesian framework provides a range of tools that allow to describe possible solutions as a distribution, quantify uncertainty and perform further tasks like hypothesis tests. We consider problems where the posterior distribution of a Bayesian imaging inverse problem has a log-concave but non-smooth density function. For the computation of the maximum a posteriori (MAP) point estimate, this would usually mean that a proximal optimization algorithm can be employed to efficiently compute a solution. While convex analysis provides solid convergence theory in this case, theory on the problem of sampling from the posterior is not as well-developed yet. 

Markov chain Monte Carlo (MCMC) methods are useful for performing Bayesian inference in this context, by allowing for sampling from complex distributions. Compared to other classes of MCMC methods, sampling schemes based on Langevin diffusion processes have proved efficient in inverse problems and imaging applications \cite{Roberts1996,Durmus2022,Durmus2019}. Langevin diffusion based MCMC sampling algorithms are all discretisations of the same stochastic differential equation (SDE) which drives Langevin diffusion processes and whose invariant distribution is the target distribution from which we want to sample. In the discrete setting, there is usually a bias between the law of the samples and the target measure, hence many works have been concerned with the characterization of the invariant distribution of the Markov chain and the bias \cite{Roberts1996,Durmus2017,Wibisono2018,Dalalyan2019}. Further questions typically concern convergence speed of the Markov chain and possible correction steps to overcome the bias and draw unbiased samples from the target. We are interested in a Langevin sampling algorithm that is based on a forward-backward discretization of the potential term in the underlying SDE. The algorithm has been analyzed under the name proximal stochastic gradient Langevin algorithm (PSGLA) in \cite{Salim2020} and as stochastic proximal gradient Langevin dynamics in \cite{Durmus2019a}. It can be viewed as a sampling equivalent to proximal gradient descent algorithms (arising from a forward-backward discretization of gradient flows), with the difference that a stochastic term is added in every iteration step in the argument to the proximal mapping. 

In many imaging inverse problems, the proximal operator of some functional of interest, e.g.\ total variation (TV) \cite{Rudin1992,Burger2013} or total generalized variation (TGV) functionals \cite{Bredies2015}, has no closed form and has to be approximated by some iterative procedure. While in the optimization setting, there have been works proving the convergence of algorithms under the assumption that proximal operators are evaluated only inexactly \cite{Alber1997,Salzo2012,Villa2013,Rasch2020}, corresponding sampling algorithms lack such theory. We consider a generalization of PSGLA in which proximal points are evaluated only inexactly up to some accuracy level. Our analysis uses techniques from convex optimization by exploiting the close relationship between sampling tasks and optimization in the corresponding space of measures.

The rest of the paper is organised as follows. \Cref{sec:setup} first introduces the problem and the relevant existing Langevin algorithms that are based on the evaluation of proximal operators. We then define the considered notion of inexactness in the proximal mappings and give our algorithm. The convergence theory is carried out in \cref{sec:convergenceTheory} by proving nonasymptotic and asymptotic convergence results. We provide numerical examples in \cref{sec:numerics}.

\subsection{Related Work}
The idea to use techniques from convex analysis to study Langevin sampling algorithms goes back to the seminal work \cite{Jordan1998}, by which Langevin dynamics corresponds to the gradient flow of relative entropy with respect to the target in the space of probability measures endowed with the Wasserstein metric. By using a coupling argument, typical bounds on some distance measure of optimization iterates along gradient flows in Euclidean space can be translated to the related sampling algorithm \cite{Durmus2019a,Vempala2019,Salim2020}. In \cite{Wibisono2018}, the author explored this correspondence further and gave explanations for the unavoidable bias of Langevin sampling algorithms.

Several works have proposed Langevin algorithms that involve the evaluation of proximal operators. In \cite{Pereyra2016,Durmus2018}, the authors approximate non-smooth potential terms by their smooth Moreau-Yosida regularization and apply the standard unadjusted Langevin algorithm (ULA) for smooth potentials to the regularized target. \cite{Bernton2018} overcomes the problem of non-smooth potentials by discretising the underlying SDE implicitly, leading to a sampling equivalent of proximal gradient descent. For a special case of the Moreau-Yosida smoothing parameter, the method in \cite{Pereyra2016} also falls in this class. As a generalization of this idea, splittings of the potential into smooth and nonsmooth terms allow to combine explicit and implicit discretizations. This results in a forward-backward type sampling method analyzed in \cite{Salim2020}. In another recent work, improved convergence behaviour using an implicit midpoint discretization of Langevin dynamics has been suggested, which also results in the evaluation of proximal mappings at each iteration \cite{Klatzer2023}.

Up to our knowledge, inexact evaluations of proximal operators with deterministic errors in Langevin sampling algorithms have not been considered yet in the literature. However, several works have been concerned with cases where stochastic estimators of potential terms are available and cases where the potential is only known inexactly in some other way. \cite{Huggins2017} analysed the continuous-time dynamics of the underlying SDE and a specific first-order approximation to computationally costly gradients in ULA. In \cite{Dalalyan2019}, the analysis of ULA was extended to versions with inexact potential gradients, where the estimators have bounded bias and variance. Stochastic but unbiased estimators of the gradient with bounded variance were considered in \cite{Salim2020,Durmus2019a} for ULA and a proximal Langevin algorithm respectively. In a recent work, \cite{Wibisono2022} proved convergence of the inexact continuous-time dynamics and ULA under the assumption of a bounded moment generating function of the gradient estimator for applications in score-based generative models. Another work that is related to our approach is \cite{Salim2019}, where the authors assume that the nonsmooth potential term allows a splitting and analyze the case of stochastic estimators of the proximal mapping.

\subsection{Contributions}
Using the notion of inexact evaluation of proximal operators that has been considered in \cite{Alber1997,Villa2013,Salzo2012,Rasch2020}, we generalize PSGLA to allow inaccurate proximal points. We show that the convergence analysis of PSGLA carried out in \cite{Salim2020} can be recovered in the inexact setting. Particularly, we recover the same convergence rates in the exact case and quantify the additional bias between the algorithm's stationary distribution and the target due to errors in the inexact case. We show that the additional bias stays bounded for a bounded sequence of errors and, if the errors go to zero in a strongly convex setting, decays to zero during the iteration as fast as the mean of the errors. Numerically, we test the algorithm by sampling from the posterior distribution of typical imaging inverse problems. In cases where the proximal operator has a closed form this allows us to compare the algorithm's exact and inexact version, and in cases where proximal points can only be approximated it showcases how proximal Langevin sampling can efficiently be applied to such imaging problems.
\section{A Langevin Sampling Algorithm with Inexact Proximal Points}\label{sec:setup}

\subsection{Problem Formulation}\label{subsec:problemformulation}
The method we analyze can be applied to the general problem of sampling from a distribution $\mu^\ast$ in finite dimensions with unknown density normalization constant. We will denote the density function by $\mu^\ast(x)$ and assume that it is of the form
\begin{equation}\label{eq:bayes_law}
	\mu^\ast(x) = \frac{1}{Z} \exp(-F(x)-G(x)),
\end{equation}
with two functions $F,G: \calX \to \bbR$, where $\calX := \bbR^d$. The sum $F+G$ is called the potential, and the normalization constant $Z$ is typically hard to compute and hence unknown.

\begin{example}
    The problem setting contains a general class of ill-posed imaging inverse problems that are modelled in the framework of Bayesian statistics. Denoting $x \in \calX = \bbR^d$ and $y \in \bbR^m$ for the observation, assume that $x,y$ have joint density $p(x,y)$ with respect to the Lebesgue measure. The acquisition process of $y$ is described by the likelihood $p(y|x)$. This is combined with a prior distribution with density $p(x)$ representing structural knowledge or assumptions about the true solution. Given an observation $y$, the central subject of analysis is the posterior distribution given by
    $$ \mu^\ast(x) = p(x|y) = \frac{p(y|x) p(x)}{Z}, $$
    with the model evidence $Z = \int p(y|\tilde{x})p(\tilde{x})\rmd \tilde{x}$. The potential terms $F$ and $G$ are chosen such that $F(x) + G(x) =  -\log p(y|x) - \log p(x) $.
\end{example}

The potential terms $F$ and $G$ are further assumed to satisfy the following criteria.

\begin{assumption}\label{assumption1}
	$F:\calX\to\bbR$ is $\lambda_F$-strongly convex, $\lambda_F \ge 0$ (which in the case $\lambda_F = 0$ means only convex) and differentiable, where $\nabla F$ is $L$-Lipschitz continuous.\\
	$G:\calX\to\overline{\bbR}:=\bbR\cup\{\infty\}$\footnote{Throughout the paper, we use the convention $\exp(-\infty)=0$.} is proper, convex and lower semicontinuous (lsc). Further, its convex conjugate, denoted by $G^\ast$, is $\lambda_{G^\ast}$-strongly convex, $\lambda_{G^\ast} \ge 0$. We allow $\lambda_{G^\ast} = 0$, in which case this is no further assumption so that $G$ is not necessarily differentiable and can take value $\infty$.
\end{assumption}

\begin{assumption}\label{assumption2}
	The function $V = F+G$ satisfies $\int \exp(-V(x)) \,\mathrm{d}x < \infty$ and $\exp(-V) \in S_{\textrm{loc}}^{1,1}(\calX)$.
\end{assumption}

Due to the first part of \cref{assumption2}, $\mu^\ast$ is a valid probability distribution on $\calX$. Since its density is log-concave by \cref{assumption1}, all its moments are finite (see \cite{Brazitikos2014}, Lemma 2.2.1), in particular $\mu^\ast$ has finite second moment, denoted $\mu^\ast \in \calP_2(\calX)$. In \cite{Salim2020}, it is shown that $F$, $G$ and $V$ are $\mu^\ast$-a.e. differentiable under \cref{assumption1,assumption2} and that the target density satisfies a first-order optimality condition on which the analysis of their algorithm is based. We further require $\nabla G$ to be an $L^2$ function w.r.t.\ $\mu^\ast$:

\begin{assumption}\label{assumption3}
	It holds $\int_{\interior(\domain(G))} \norm{\nabla G(x)}^2 \,\mathrm{d}\mu^\ast(x) < \infty$.
\end{assumption}

Due to \cref{assumption1} we only consider cases where the prior and the likelihood are both log-concave, which covers many relevant noise models and typical imaging priors. In the log-concave case the computation of the maximum a posteriori (MAP) point estimate
$$ x_{\mathrm{MAP}} := \argmax_x \mu^\ast(x) = \argmin_x \left\{ F(x) + G(x) \right\} $$ 
can be carried out efficiently and with well-developed theory on convergence behaviour using convex optimization algorithms \cite{Chambolle2016,Combettes2011}. However, for more advanced statistical tasks in which the point estimate $x_{\mathrm{MAP}}$ can not provide sufficient information about the posterior, it is often necessary to draw representative samples from the distribution $\mu^\ast$.

In imaging inverse problems with regularizers that enforce some kind of sparsity, while the log-likelihood might be smooth, the prior log-density is often not differentiable w.r.t.\ the Lebesgue measure, so that $F(x) =  -\log p(y|x)$ and $G(x) = -\log p(x)$ is the logical choice. We mention here though that the terms in the posterior log-density can be split up differently into a smooth term $F$ and a potentially non-smooth term $G$ as long as $F$ and $G$ satisfy the assumptions above.

\subsection{Existing Proximal Langevin Sampling Schemes}\label{subsec:langevinmcmc}
Suppose we want to sample from a target distribution $\mu^\ast$ with density proportional to $\exp(-V)$ for a potential $V:\bbR^d \to \bbR$ which for now is assumed to be differentiable. Langevin diffusion processes for this target are solutions of the Ito stochastic differential equation
\begin{equation}\label{eq:LangevinSDE}
    \rmd X_t = -\nabla V(X_t) \rmd t + \sqrt{2} \rmd W_t,
\end{equation}
where $W_t$ is a Wiener process in $\bbR^d$. $\mu^\ast$ is the unique invariant probability measure of the Markov semigroup associated with this SDE. Furthermore, if $\mu^\ast$ satisfies some regularity assumption like the logarithmic Sobolev inequality, every solution of the SDE converges exponentially fast in time to the stationary target $\mu^\ast$ \cite{Roberts1996,Vempala2019}. Since the SDE only has an explicit solution for specific cases of the function $V$, sampling algorithms based on Langevin diffusion are usually time-discretized Markov chains that approximate continuous processes which solve the SDE.

A straightforward Euler-Maruyama discretization of \eqref{eq:LangevinSDE} leads to the following sampling scheme called unadjusted Langevin algorithm (ULA) given by
\begin{equation}\label{eq:ULA}
    X^{k+1} = X^k - \gamma \nabla V (X^k) + \sqrt{2\gamma}\, \xi^k,\quad \xi^k \sim \rmN(0,I_d),
\end{equation}
with step size $\gamma > 0$. In the case of strongly convex $V$, the algorithm is well-understood in non-asymptotic and asymptotic behaviour of the law $\mu^k$ of iterate $X^k$ \cite{Roberts1996,Dalalyan2019,Durmus2017}.If the step size $\gamma$ is small enough, the Wasserstein distance between $\mu^k$ and $\mu^\ast$ decreases exponentially with $k$ up to a constant bias scaling like the square root of the step size $\sqrt{\gamma}$ \cite{Dalalyan2019} and similar results hold for the measures' total variation distance. The bias is usually unavoidable in Langevin diffusion based sampling algorithms if there is no additional step correcting for it. It is attributed to the fact that the iteration is performing an unbalanced discretization in time of the gradient flow of relative entropy. The gradient step in $V$ is a time-discrete step in the expected potential value while the addition of a normal random variable solves the gradient flow of negative entropy in continuous time, see \cite{Wibisono2018}. For variable step sizes decaying to zero at the right rate, convergence in different metrics (Wasserstein distance, total variation distance or KL-divergence) can be ensured, with accuracy $\epsilon$ in Wasserstein distance reached after at most $\mathcal{O}(d\epsilon^{-2})$ iterations in the case of strongly convex $V$. If $V$ is only convex, weaker results hold with the number of iterations still depending at most polynomially on $d$ \cite{Dalalyan2019,Durmus2022}. Particularly due to this moderate dependence on the dimension $d$ in comparison to other MCMC algorithms, sampling schemes based on Langevin diffusion are very efficient in high-dimensional applications and have gained popularity in recent years.

In this work, we are interested in the case when $V=F+G$ contains a non-smooth term $G$ which makes ULA not well-defined. Technically, a variant of ULA could still be used if $G$ is supported everywhere and subdifferentiable by replacing $\nabla V$ with an element of the subdifferential $\partial V$. However, the theoretical convergence guarantees then do not hold anymore unless other restrictive conditions are satisfied, e.g.\ Lipschitz-continuity of the nonsmooth potential~\cite{Durmus2019a}. In the literature, several other strategies to circumvent the problem of nonsmoothness have been proposed.

One popular technique relies on changing the target density by replacing the non-smooth part $G$ of the potential with its regularized Moreau--Yosida envelope $G^\lambda$ defined by $G^\lambda(x) := \argmin_{y}\left\{ G(y) + \frac{1}{2\lambda} \lVert x-y \rVert_2^2 \right\}$. ULA can then be applied to the altered target measure $\mu^\lambda$ with density $\exp(-F(x)-G^\lambda(x))$. The resulting algorithm MYULA (Moreau--Yosida Unadjusted Langevin Algorithm) is given by
\begin{equation*}
    X^{k+1} = \left( 1-\frac{\gamma}{\lambda} \right) X^k - \gamma \nabla F(X^k) + \frac{\gamma}{\lambda} \prox_{\lambda G}(X^k) + \sqrt{2\gamma}\,\xi^k,\quad \xi^k \sim \rmN(0,I_d).
\end{equation*}
Since $\nabla G^\lambda$ is Lipschitz continuous, the convergence theory of ULA can be applied to show that the distribution of generated samples is close to $\mu^\lambda$. If $\lambda$ is small enough, the total variation distance between $\mu^\lambda$ and $\mu^\ast$ is also small, allowing a convergence theory for MYULA \cite{Durmus2018,Durmus2022}. Since the Moreau-Yosida envelope changes the target distribution, a forward-backward envelope which changes the distribution, but preserves the MAP and shares connections with the method investigated in the present paper has been considered in \cite{Eftekhari2023}.

In this work, we consider a different approach to the problem of non-smooth potentials. Instead of applying an entirely explicit discretization to \eqref{eq:LangevinSDE}, the smooth part $F$ of the potential is discretized explicitly and the non-smooth part $G$ implicitly. Analogous to forward-backward splitting algorithms in optimization, we can design a sampling algorithm including a forward step, the additive stochastic term (diffusion) and a backward step:
\begin{equation*}
    X^{k+1} = \prox_{\gamma G}(X^k - \gamma \nabla F(X^k) + \sqrt{2\gamma}\,\xi^k),\quad \xi^k \sim \rmN(0,I_d).
\end{equation*}
This algorithm is analyzed in \cite{Salim2020} as proximal stochastic gradient Langevin algorithm (PSGLA) with the addition that the gradients $\nabla F$ are replaced by unbiased estimators with bounded variance. In the present paper, we consider this algorithm in the setting where the evaluation of the proximal mapping is carried out inexactly. The gradients $\nabla F$ are evaluated exactly, so that we call this scheme proximal gradient Langevin algorithm (PGLA) here. Note, however, that the whole analysis can be extended without major modification to the stochastic gradient case of \cite{Salim2020}.

The two schemes MYULA and PGLA are comparable in the sense that both algorithms evaluate the gradient of $F$ and a proximal mapping of $G$ in each iteration to overcome the problem of non-smooth $G$. To draw the link to optimization algorithms, PGLA can be viewed as a sampling equivalent of forward-backward splitting where the diffusion step is carried out between the two discrete time steps along the potential fields of $F$ and $G$. MYULA, on the other hand, acts like the sampling correspondence of gradient descent on the partly Moreau--Yosida regularized objective $F + G^\lambda$. The flexbile smoothing parameter $\lambda$ in MYULA allows to trade off the algorithm's speed to the approximation of the true target. PGLA, however, has the advantage that constraints encoded in the potential $G$ are not relaxed and the samples will always remain in the support of $G$.

Another algorithm that can be related to MYULA has been considered in \cite{Laumont2022}. In a manner akin to Plug \& Play optimization schemes, the considered algorithm therein replaces the proximal mapping by a trained denoiser. This approach allows the analysis to be extended to prior distributions that are implicitly defined by the denoiser, without requiring access to a density and is more agnostic with regards to the true statistical model of observed data. In the PGLA setting, we assume to know the potentials exactly and only evaluate the proximal mapping inexactly, but it could be interesting future work to relax this assumption to cases where the inexactness in the proximal mapping is related to an inexactness in the underlying model $G$ (implicitly caused by, e.g., training a denoiser).

We mention here briefly that further discretizations are possible in the non-smooth case, e.g. by smoothing the whole potential term, or by discretizing to only a single backward step in the potential followed by the step along the gradient flow of negative entropy. The resulting methods have also been called proximal Langevin algorithms \cite{Bernton2018,Pereyra2016} and are special instances of MYULA, PGLA or both of the latter. An algorithm arising from an implicit midpoint discretization, which also results in the computation of proximal points at each step, has recently been analyzed in \cite{Klatzer2023}.

\subsection{A Notion of Inexactness for Proximal Points}\label{subsec:inexactproximalpoints}
We want to generalize PGLA by allowing an inexact evaluation of the proximal mapping in every iteration. The distribution of the resulting Markov chain is, as typical for Langevin sampling algorithms, a biased approximation to the target density. 

We begin by defining what we mean by an inexact evaluation of the proximal mapping. This first requires us to introduce the $\epsilon$-subdifferential, see \cite{Zalinescu2002}, which is a generalization to the subdifferential of a function. Given $G:\calX \to \overline{\bbR}$ for some Hilbert space $\calX$ and an $\epsilon \ge 0$, it is defined by
$$ \partial_\epsilon G (x) = \left\{p \in \calX\,:\, G(y) \ge G(x) + \inpro{p}{y-x} - \epsilon \,\ \forall y \in \calX \right\}, $$
where $\inpro{\cdot}{\cdot}$ is the inner product in $\calX$. Note that when $\epsilon = 0$, this reduces to the definition of the subdifferential $\partial G(x)$. Further, it holds $\partial G(x) \subseteq \partial_{\epsilon_1}G(x) \subseteq \partial_{\epsilon_2}G(x)$ for all $0 < \epsilon_1 < \epsilon_2$. Using the $\epsilon$-subdifferential, we now define inexact proximal points.
\begin{definition}[\cite{Rasch2020}]\label{def:inexactProx}
	For $G:\calX \to \bbR$, $\gamma > 0$ and $\epsilon \ge 0$, say that $y \in \calX$ is an $\epsilon$-approximation of the proximal point $\prox_{\gamma G}(x)$ if $y$ satisfies
	\begin{equation}\label{eq:dfn-inexact-subgradient}
	    y \approx^\epsilon \prox_{\gamma G}(x) \ \Leftrightarrow\ \frac{x-y}{\gamma} \in \partial_\epsilon G (y).
	\end{equation}
 It is worth noting that for $\epsilon > 0$, the operator $x \mapsto \partial_\epsilon G(x)$ is not monotone, so that $\epsilon$-inexact proximal points are generally not unique (in contrast to the case $\epsilon = 0$, see \cite{Bauschke2011}, chap. 23).
\end{definition}
We illustrate the inexact subdifferential and proximal mapping on two simple examples.
\begin{example}
    Let $G(x) = \frac{1}{2\gamma}\norm{x-z}_2^2$, $x \in \bbR^d$. The inexact subdifferential of $G$ is
    $$ \partial_{\epsilon} G(x) =  \left\{ \frac{x-z-r}{\gamma} \ : \ \norm{r} \le \sqrt{2\gamma\epsilon} \right\},$$
    which shows that the set of inexact proximal points is given by
    $$ x \approx^\epsilon \prox_{\tau G}(y) \ \Leftrightarrow \ x \in \left\{ \frac{\gamma y + \tau (z + r)}{\gamma + \tau} \ : \ \norm{r} \le \sqrt{2\gamma\epsilon} \right\}. $$
\end{example}
\begin{example}
    Let $G(x) = |x|$, $x \in \bbR$ and $\epsilon \ge 0$. Then from the definition of the inexact subdifferential, one can compute that 
    $$\partial_{\epsilon}G(x) = \left\{ \begin{array}{ll}
        \left[\max(1-\frac \epsilon x, -1), 1\right] & \textrm{if } x>0, \\
        \left[-1, 1\right] & \textrm{if } x=0, \\
        \left[-1, \min(-1-\frac\epsilon x, 1)\right] & \textrm{if } x<0.
    \end{array}\right.$$
    Note that the inexact subdifferential is not a singleton for all $x \in \bbR$ if $\epsilon > 0$. By solving the right side of \eqref{eq:dfn-inexact-subgradient} for $y$, one can compute that for $\epsilon > 0$ the inexact proximal mapping is given by
    $$ y \approx^\epsilon \prox_{\tau |\cdot |}(x) \ \Leftrightarrow \ y \in \left[ \max(-h_\tau^\epsilon(-x), x-\tau), \min(h_\tau^\epsilon(x),x+\tau) \right], $$
    where $h_\tau^\epsilon(x) := (x-\tau)/2 + \sqrt{(x-\tau)^2/4 + \epsilon \tau}$. The inexact subgradients and proximal points are illustrated in \cref{fig:inexact_prox_absolute_value}.
\end{example}
\begin{figure}[t]
    \centering
    \includegraphics[width=\linewidth]{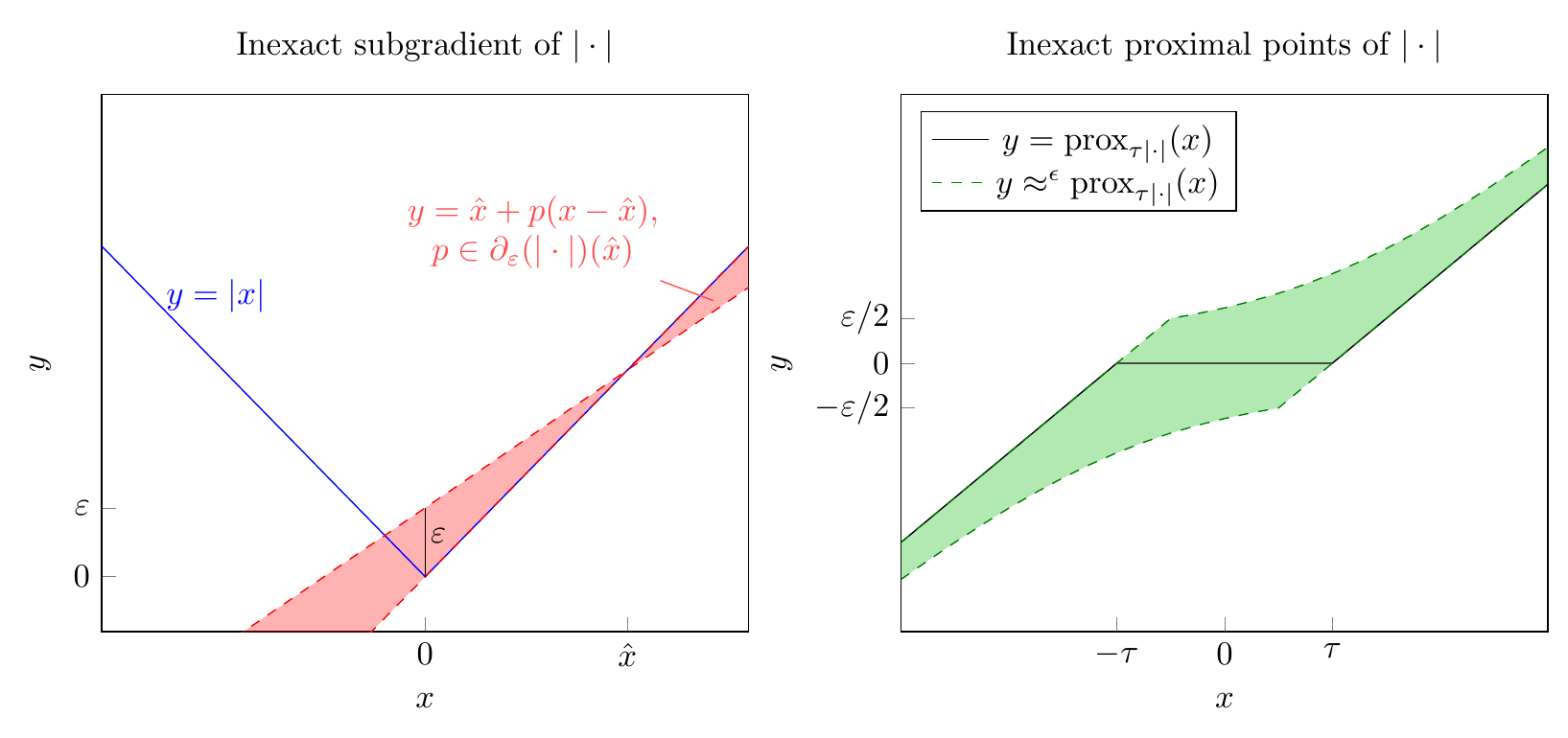}
    \caption{Illustration of the inexact subgradient of the absolute value function and the corresponding inexact proximal points.}
    \label{fig:inexact_prox_absolute_value}
\end{figure}
The present definition of inexact proximal points has been used in \cite{Alber1997} to analyze inexact proximal gradient descent and also under the name ``type-2 approximation'' in \cite{Salzo2012,Villa2013,Rasch2020} where several types of inexact proximal points were used to analyze inexactness in accelerated forward-backward splitting methods and the primal dual hybrid gradient algorithm. In \cref{sec:convergenceTheory}, we comment on interpretation of the definition used here, how it can be verified efficiently in numerical practice and how it relates to others considered in the literature.

\subsection{The Proposed Inexact Sampling Scheme}
With the definition of inexact evaluation of proximal mappings, we are able to generalize the proximal Langevin algorithm considered in \cite{Salim2020}.

\begin{algorithm}[H]
\caption{Inexact PGLA}\label{algo:iPGLA}
\begin{algorithmic}
    \STATE{Input: $X^0,K,\{\gamma_k\}_{k=0}^{K-1},\{\epsilon_k\}_{k=0}^{K-1}$}
    \FOR{$k = 0, \dots, K-1$}
    \STATE{Draw $\xi^{k+1} \sim \mathrm{N}(0,I_d)$}
    \STATE{Compute $X^{k+1} \approx^{\epsilon_k} \prox_{\gamma_k G}\left( X^k - \gamma_k \nabla F(X^k) + \sqrt{2\gamma_k}\,\xi^k \right)$}
    \ENDFOR
\end{algorithmic}
\end{algorithm}

We allow both the error level $\epsilon_k$ and the step size $\gamma_k$ to change over time. The convergence of the algorithm is analyzed in \cref{sec:convergenceTheory} for both fixed and flexible step size choices and errors. For the sake of notation we split the iteration into three intermediate steps of the form
\begin{align*}
    X^{k+1/3} &= X^k - \gamma_k \nabla F(X^k),\\
    X^{k+2/3} &= X^{k+1/3} + \sqrt{2\gamma_k}\,\xi^k,\\
    X^{k+1} &=  S_k(X^{k+2/3}) \approx^{\epsilon_k} \prox_{\gamma_k G}\left(X^{k+2/3} \right),
\end{align*}
where $S_k:\calX \to \calX$ is the operator that maps to the $\epsilon_k$-inexact proximal point that our algorithm chooses, since inexact proximal points are in general non-unique. Denote $\mu^j = \mathrm{Law}(X^j), j \in \bbN/3$ for the distributions of the samples at each step. Note that for $\epsilon_k=0$, we of course recover the iteration formula of the original PGLA algorithm.
\section{Convergence Theory}\label{sec:convergenceTheory}
In this section, we give a short overview of the relation between the analysis of optimization and sampling algorithms. We then give auxiliary results on the type of inexact proximal points which we use in \cref{algo:iPGLA} and comment on how to numerically ensure the inexact computation for a given accuracy level. This allows us to prove the main nonasymptotic convergence result \cref{thm:nonasymptotic_result_type2error} and two unrolled versions that lead to asymptotic rates in \cref{thm:asymptotic_result_type2error_fixedSteps,thm:asymptotic_result_type2error_decaying_stepsizes} for the cases of fixed or decaying step sizes and proximal mapping accuracy levels.
%
%
%
%
%
%
%
\subsection{Sampling as Optimization in the Wasserstein Space}
In several recent works, Langevin Monte Carlo algorithms were successfully analyzed using techniques from optimization theory \cite{Durmus2019a,Vempala2019,Salim2020,Wibisono2018}. We take the same perspective here and motivate the formulation of the sampling task as an optimization problem. This then leads to a reformulation as a saddle point problem following the arguments in \cite{Salim2020}, which is the basis of our convergence analysis. In the following, let $\calP_2(\calX)$ be the set of probability measures on $\calX = \bbR^d$ with finite variance. Let the Wasserstein $p$-metric be defined as usual for $0<p<\infty$ by 
\begin{equation*}\label{eq:wassersteindist}
    \calW_p(\mu,\nu) := \left(\inf_{\pi} \bbE_{(X,Y)\sim\pi} \left[ \lVert X-Y \rVert^p \right] \right)^{1/p},
\end{equation*}
where the infimum is taken over all possible couplings $\pi$ of the measures $\mu$ and $\nu$, i.e., all probability measures $\pi \in \calP_2(\calX^2)$ with marginal distributions $\mu$ and $\nu$. The sampling task then corresponds to the problem of approximating $\mu^\ast$ in the metric space $(\calP_2(\calX), \calW_2)$.

For any $\mu \in \calP_2(\calX)$, define the negative entropy
$$\calH (\mu) = \int \log(\mu(x)) \rmd \mu(x)$$
with $\calH(\mu) := \infty$ if $\mu$ has no density w.r.t.\ the Lebesgue measure. Define the potential energy functionals for the functions $F$ and $G$ by
$$\calE_F(\mu) = \int F(x) \rmd \mu(x), \quad \calE_G(\mu) = \int G(x) \rmd \mu(x). $$
For $\mu, \nu \in \calP_2(\calX)$ with $\mu \ll \nu$ ($\mu$ absolutely continuous with respect to $\nu$) we further define the Kullback--Leibler divergence of $\mu$ from $\nu$
$$ \KL(\mu,\nu) = \int \rmderiv{\mu}{\nu}(x) \log\left(\rmderiv{\mu}{\nu}(x)\right) \rmd \nu (x), $$
which is often also called relative entropy of $\mu$ w.r.t.\ $\nu$. It can easily be seen that with $\mu^\ast$ defined as in \eqref{eq:bayes_law} it holds $\KL(\mu, \mu^\ast) = \calH(\mu) + \calE_F(\mu) + \calE_G(\mu)$ for all $\mu \in \calP_2(\calX)$ \cite{Ambrosio2008}. Since the Kullback--Leibler divergence obeys $\KL(\mu,\mu^\ast) \ge 0$ and further $\KL(\mu,\mu^\ast) = 0$ if and only if $\mu = \mu^\ast$, the task of approximating $\mu^\ast$ corresponds to minimizing the free energy functional $\calH + \calE_F + \calE_G$. 

Exploring this correspondence further provides intuition why Langevin dynamics in the form \eqref{eq:LangevinSDE} is a good choice when we want to sample from $\mu^\ast$. In \cite{Jordan1998}, it was shown that the gradient flow of $\KL(\mu,\mu^\ast)$ in $\calP_2(\calX)$ equipped with the Wasserstein 2-distance actually corresponds to the Fokker-Planck equation of the Langevin SDE \eqref{eq:LangevinSDE}. Hence by approximating Langevin diffusion processes, we sample approximately from solutions of a Fokker-Planck equation which we can hope to converge to $\mu^\ast$ quickly. By exploiting this relationship theoretically, it is possible to derive convergence results in Wasserstein distance to $\mu^\ast$ of diffusion processes driven by Langevin dynamics.

The authors in \cite{Salim2020} analyzed \cref{algo:iPGLA} by deriving a first order optimality condition of the minimization problem 
\begin{equation}\label{eq:opt-problem-kl}
    \min_{\mu \in \calP_2(\calX)} \KL(\mu,\mu^\ast) = \min_{\mu \in \calP_2(\calX)} \calH(\mu) + \calE_F(\mu) + \calE_G(\mu)
\end{equation}%
As in optimization problems in Euclidean space, this allows to transform the optimization problem in the primal variable $\mu$ into a saddle point problem. The authors of \cite{Salim2020} proved a saddle point result for the generalized Lagrangian. In our analysis, we will only use that the resulting primal-dual gap is positive, so that we simply introduce the necessary notation here and refer the reader to \cite{Salim2020} for details and the proof of the saddle point lemma.
By $T_{\mu_1\to\mu_2}:\calX \to \calX$, we will denote the optimal transport map from $\mu_1 \in \calP_2(\calX)$ to $\mu_2 \in \calP_2(\calX)$, i.e.\ the map for which the transport plan $(I,T_{\mu_1\to\mu_2})\# \mu_1$ realizes the infimum in the definition of $\calW_2(\mu_1,\mu_2)$  \cite{Villani2003}. We introduce a dual variable $\psi \in L^2(\mu^\ast;\calX)$, which defines a coupling between the random variable $X$ and a new random variable $Y$ in the following way: For given $(\mu,\psi)$, the joint law of $(X,Y)$ is given by $\pi(\mu,\psi) := (T_{\mu^\ast\to\mu},\psi)\# \mu^\ast$. We will further sometimes denote the marginal distribution of $Y$ by $\nu(\psi) := \psi \# \mu^\ast$. For any $\mu \in \calP_2(\calX)$ and any $\psi \in L^2(\mu^\ast;\calX)$, a generalized Lagrangian is defined by
$$ \calL(\mu, \psi) = \calH(\mu) + \calE_F(\mu) - \calE_{G^\ast}(\nu(\psi)) + \bbE_{(X,Y)\sim \pi(\mu,\psi)}[\inpro{X}{Y}], $$
and a corresponding duality gap by
$$\calD(\mu, \psi) := \calL(\mu, \psi^\ast) - \calL(\mu^\ast, \psi), $$
where $\psi^\ast = \nabla G$ is the optimal value of the dual variable. The name `duality gap' is justified by the following result. 
\begin{lemma}[\cite{Salim2020}]
Let \cref{assumption1,assumption2,assumption3} hold. For every $\mu \in \calP_2(\calX)$, $\psi \in L^2(\mu^\ast;\calX)$ it holds $\calD(\mu, \psi) \ge 0$ and $\calL(\mu, \psi) \le \KL(\mu,\mu^\ast)$. The pair $(\mu^\ast, \psi^\ast)$ is a saddle point of $\calL$ with value 0 in the sense that for all $\mu \in \calP_2(\calX)$, $\psi \in L^2(\mu^\ast;\calX)$ it holds
$$ \calL(\mu^\ast, \psi) \le 0 = \calL(\mu^\ast,\psi^\ast) \le \calL(\mu, \psi^\ast), $$
and further $\calL(\mu^\ast, \psi) = 0$ if and only if $\psi = \psi^\ast$ holds $\mu^\ast$-a.e.
\end{lemma}

The introduction of the dual variable allows the analysis of $\calD$ along the iterates in \cref{algo:iPGLA}. In the discrete setting of \cref{algo:iPGLA}, we can analogously define dual samples
$$Y^{k+1} := \frac{X^{k+2/3}-X^{k+1}}{\gamma_k} = \frac{1}{\gamma_k} (I-S_k) (X^{k+2/3}).$$
Denoting $\psi^{k+1} := (I - S_k)/{\gamma_k} \circ T_{\mu^\ast \to \mu^{k+2/3}}$ where we assume $S_k \in L^2(\mu^\ast,\calX)$, we are interested in the distribution of the dual variable $\nu^{k+1} := \mathrm{Law}(Y^{k+1}) = \psi^{k+1} \#\mu^{\ast}$. Note that the resulting maps $\psi^{k+1}$ can be understood as approximations of the optimal dual variable $\psi^\ast = \nabla G$ resulting from a time discretization of the gradient flow of $G$ with step size $\gamma_k$.
%
%
%
%
%
%
%
%
\subsection{Auxiliary Results on Inexact Proximal Points}\label{subsec:inexact_prox_mappings}
Let in the following $G:\calX \to \bar\bbR$. We start by stating some technical results about $\epsilon$-subdifferentials and our notion of inexact proximal mappings from \cref{def:inexactProx}.
\begin{lemma}[see Thm 2.4.4 (iv) in \cite{Zalinescu2002}]\label{lem:convex_duality_epsilon_subdifferential}
Let $\epsilon \ge 0$. If $G$ is convex, proper and lsc, then $$p \in \partial_\epsilon G(u)\,\Leftrightarrow\, u \in \partial_\epsilon G^\ast(p).$$
\end{lemma}
The next Lemma allows us to characterize the $\epsilon$-subdifferentials of a sum of a convex and a quadratic function.
\begin{lemma}[see Lemma 1 in \cite{Salzo2012}]\label{lem:eps_subgradient_sum_convex_quadratic}
Let $\phi_\gamma (z) := \norm{y-z}^2/(2\gamma)$ and $G_\gamma(z) = G(z) + \phi_\gamma(z)$. Then for any $\epsilon \ge 0$ it holds
\begin{align*}
    \partial_\epsilon G_\gamma(z) &= \bigcup_{\epsilon_1,\epsilon_2 \ge 0\,:\,\epsilon_1+\epsilon_2 = \epsilon} \partial_{\epsilon_1} G(z) + \partial_{\epsilon_2} \phi_\gamma (z)\\
    &= \bigcup_{\epsilon_1,\epsilon_2 \ge 0\,:\,\epsilon_1+\epsilon_2 = \epsilon} \partial_{\epsilon_1} G(z) + \left\{ \frac{z-y-r}{\gamma} \,:\, \norm{r} \le \sqrt{2 \gamma \epsilon_2} \right\}.
\end{align*}
\end{lemma}

Next, we prove an inequality resembling a typical property of subgradients of strongly convex functions, similar computations have been considered in \cite{Rasch2020}. Note however that the additional parameter $\theta \in [0,1)$ gives a slightly worse quadratic term due to the errors.
	
\begin{lemma}\label{lem:strong_convexity_epsilon_subdifferential}
Let $\epsilon \ge 0$, $G$ proper, lsc and $\lambda$-strongly convex, $\lambda \ge 0$, $u\in \mathcal{X}$, $p \in \partial_\epsilon G(u)$. Then for any $\theta \in [0,1)$ and $ v \in \mathcal{X}$ it holds
$$ G(v) \ge G(u) + \inpro{p}{v-u} + \frac{\theta\lambda}{2} \norm{u-v}^2 - \frac{\epsilon}{1-\theta}.$$
\end{lemma}

Before proving \cref{lem:strong_convexity_epsilon_subdifferential}, we remark the following.
\begin{remark}
Note that the inequality in \cref{lem:strong_convexity_epsilon_subdifferential} is in general false when the respective factors $\theta$ and $(1-\theta)^{-1}$ are omitted and replaced by one. As a counterexample consider the $\lambda$-strongly convex function $G_\lambda(v) = \frac{\lambda}{2}\norm{v}^2$ with $\partial_{\epsilon}G_\lambda(0) = \{p\,:\,\norm{p}^2\le 2\lambda\epsilon \}$. Then it holds $G_\lambda(0)-G_\lambda(v)+\inpro{p}{v-0}+\frac{\lambda}{2}\norm{0-v}^2 = \inpro{p}{v}$ which cannot be bounded for all $v\in \calX$ if $p \neq 0$.\\
From the same example, it can also be seen that the inequality in \cref{lem:strong_convexity_epsilon_subdifferential} is indeed sharp in the following sense: When the inexact subgradient obeys $\norm{p}^2 = 2\lambda\epsilon$ (so that $p\in\partial_\epsilon G_\lambda(0)$, but $p\notin\partial_{ \delta} G_\lambda(0)$ for any $\delta < \epsilon$, then the inequality becomes an equality for any pair of values $\theta \in [0,1)$ and $v = \frac{1}{(1-\theta)\lambda}p$.

Note further that when $\epsilon = 0$, the last term in the inequality of the lemma vanishes and by continuity the inequality without the last term then also holds for $\theta = 1$, which is a result frequently used when analysing subdifferentials of strongly convex functions \cite{Bauschke2011}.
\end{remark}

\begin{proof}
Since $G$ is $\lambda$-strongly convex, there exists a convex function $g$ with $G(v) = g(v) + \frac{\lambda}{2}\norm{v-u}^2$. By the characterisation of \cref{lem:eps_subgradient_sum_convex_quadratic}, we have
\begin{equation*}
    \partial_\epsilon G(u) = \bigcup_{\epsilon_1,\epsilon_2 \ge 0\,:\,\epsilon_1+\epsilon_2 = \epsilon} \partial_{\epsilon_1}g(u) + \left\{ q\,:\, \norm{q}^2 \le 2 \lambda \epsilon_2 \right\}.
\end{equation*}
Hence there exist $\epsilon_1, \epsilon_2 \ge 0$, $\epsilon_1+\epsilon_2 = \epsilon$ such that $p = p_1 + p_2$ for some $p_1 \in \partial_{\epsilon_1} g(u)$ and $p_2$ with $\norm{p_2}^2 \le 2 \lambda \epsilon_2$. We obtain that
\begin{align*}
    G(u) - G(v) + \inpro{p}{v-u} + \frac{\theta\lambda}{2}\norm{u-v}^2 &= g(u) - g(v) + \inpro{p_1+p_2}{v-u} - \frac{(1-\theta)\lambda}{2}\norm{u-v}^2\\
    &\le \epsilon_1 + \inpro{p_2}{v-u} - \frac{(1-\theta)\lambda}{2}\norm{u-v}^2\\
    &\le \epsilon_1 + \frac{1}{2(1-\theta)\lambda} \norm{p_2}^2 \le \epsilon_1 + \frac{\epsilon_2}{1-\theta} \le \frac{\epsilon}{1-\theta}.
\end{align*}
where we used the definition of the $\epsilon_1$-subdifferential and the Cauchy-Schwarz inequality followed by the rescaled Young's inequality $ab \le \frac{a^2}{2(1-\theta)\lambda} + \frac{(1-\theta)\lambda b^2}{2} $.
\end{proof}

In the computation of the proximal point $\prox_{\gamma G}(y)$ we solve problems of the form
$$\argmin_x \left\{G(x) + \frac{1}{2\gamma}\norm{x-y}^2\right\} =: \argmin_x G_\gamma(x). $$
The corresponding optimality condition is 
\begin{equation}\label{eq:optimality_conditions_proximal_point}
	0 \in \partial G_\gamma(x) \Leftrightarrow \frac{y-x}{\gamma} \in \partial G(x).
\end{equation}
Notions of inexact proximal points can be defined by relaxing these optimality conditions when the subdifferentials are replaced by $\epsilon$-subdifferentials. Doing this in the right hand side of \eqref{eq:optimality_conditions_proximal_point} leads to \cref{def:inexactProx}. When the subdifferential on the left hand side in \eqref{eq:optimality_conditions_proximal_point} is replaced, i.e. if the criterion reads $0 \in \partial_{\epsilon} G_\gamma (x)$, then the condition is less strict. In \cite{Rasch2020}, points $x$ satisfying the latter condition are called "type-1 approximations". By  \cref{lem:eps_subgradient_sum_convex_quadratic}, if $0 \in \partial_{\epsilon} G_\gamma (x)$, there exist $\epsilon_1, \epsilon_2 \ge 0$, $\epsilon_1 + \epsilon_2 = \epsilon$ and $r$ with $\norm{r} \le \sqrt{2\gamma \epsilon_2}$ such that 
$$\frac{y-x-r}{\gamma} \in \partial_{\epsilon_1} G(x). $$ 
Hence, as is pointed out also in \cite{Rasch2020}, a type-1 approximation can be seen as a type-2 approximation of accuracy $\epsilon_1$ to the proximal mapping evaluated at the erroneous input $y-r$. In particular, a type-2 approximation corresponds to $r = 0$  and is therefore always a type-1 approximation.

We only consider type-2 approximations here. This is more restrictive, but actually the type-2 error can numerically be ensured very well for a lot of common problems by keeping track of the duality gap during the computation of the proximal mapping. Assume without loss of generality $G(x) = H(Bx)$ for some proper, convex and lsc $H: \calY \to \overline{\bbR}$ and some bounded linear operator $B : \calX \to \calY$. Then the computation of $\prox_{\gamma G}(y)$ requires solving
\begin{equation}\label{eq:prox_problem_primal_form}
	 \min_x G_\gamma(x) = \min_x \left\{ H(Bx) + \frac{1}{2\gamma}\norm{x-y}_2^2 \right\}.
\end{equation}
If there exists $x_0$ such that $H$ is continuous in $Bx_0$, then by strong duality it holds \cite{Bauschke2011}
\begin{equation}\label{eq:prox_problem_dual_form}
	 \min_x G_\gamma(x) = - \min_z \left\{ \frac{\gamma}{2} \norm{ B^\ast z}^2_2 - \inpro{ B^\ast z}{y} + H^\ast(z) \right\} =: - \min_z W_\gamma(z).
\end{equation}
Under strong duality the proximal point $\hat{x}$ is further given by the dual solution $\hat{z}$ through the optimality condition $\hat{x} = y - \gamma B^\ast \hat{z}$ and the duality gap 
\begin{equation}\label{eq:duality_gap}
    \calG(x,z) = G_\gamma(x) + W_\gamma(z)
\end{equation}
vanishes only at the optimum $(\hat{x},\hat{z})$. The following lemma generalizes this to $\epsilon$-approximate solutions.
\begin{lemma}[see Proposition 2.3 in \cite{Villa2013}]\label{lem:guarantee_epsilon_approximation_by_duality_gap}
If strong duality holds, then we have $\calG(x,z) = G_\gamma(x) + W_\gamma(z) = 0$ if and only if $x,z$ are the optimal values $\hat{x} = y - \gamma B^\ast \hat{z}$. More generally, it holds
$$ \calG(y-\gamma B^\ast z,z) \le \epsilon \ \Rightarrow \ y-\gamma B^\ast z \approx^\epsilon \hat{x} = \prox_{\gamma G}(y). $$
\end{lemma}
Since the duality gap can usually be computed with little additional computational overhead, the type-2 approximation of proximal points can in practice be ensured efficiently. The next example demonstrates this procedure on a typical imaging problem.
\begin{example}\label{ex:epsilon_prox_computation_for_TV}
    A typical choice of model-based prior in imaging inverse problems is based on total variation (TV) regularization. The (isotropic) TV functional is defined by $\TV(x) := \sum_{i,j} \sqrt{(\Delta_{h}x)_{i,j}^2 + (\Delta_{v}x)_{i,j}^2}$ where $\Delta_h, \Delta_v$ denote horizontal and vertical finite differences in the pixelated image $x$. Computing the proximal mapping of TV requires the solution of
    $$ \prox_{\gamma \TV}(y) = \argmin_x \left\{ \norm{(\Delta_h,\Delta_v)x}_{2,1} + \frac{1}{2\gamma}\norm{x-y}_2^2 \right\}, $$
    which is an instance of \eqref{eq:prox_problem_primal_form} with $B=(\Delta_h,\Delta_v)$ and $H = \norm{\cdot}_{2,1}$, which is the nested $l^2$ norm over horizontal and vertical component and $l^1$ norm over the pixels. The solution of this problem (also called ROF problem \cite{Rudin1992}) has no closed form, so its computation is usually carried out iteratively, which only gives an approximation instead of the exact point. A typical approach (see \cite{Chambolle2004}) consists in solving the dual problem \eqref{eq:prox_problem_dual_form}, which in this case reads
    \begin{equation}\label{eq:dual_form_ROF_problem}
        \hat{z} = \argmin_{z\,:\,\norm{z}_{2,\infty} \le \gamma} \left\{ \frac{\gamma}{2} \norm{B^\ast z}^2_2 - \inpro{B^\ast z}{y} \right\},
    \end{equation}
    where $B^\ast = (\Delta_h^\ast,\Delta_v^\ast)$. The dual solution $\hat{z}$ can be approximated by a sequence of dual iterates $z^n$ using a first order optimization method, e.g. accelerated proximal gradient descent (AGD), see \cite{Beck2009,Chambolle2016}. If the computation is stopped once $\calG(y-\gamma B^\ast z^n,z^n) \le \epsilon$ is satisfied, then by \cref{lem:guarantee_epsilon_approximation_by_duality_gap} the image $y-\gamma B^\ast z^n$ is an $\epsilon$-approximation to the exact proximal point. 
\end{example}
\begin{remark}
    Another flavour of total variation can be considered by exchanging the inner $l^2$ norm in the definition of TV with anisotropic variants. A practically relevant example is $\TV_{1}(x) := \norm{(\Delta_h,\Delta_v)x}_{1,1}$. The functional $\TV_{1}$ allows a splitting into separate, one-dimensional TV functionals along the rows and columns of the image. Since the proximal point of 1D TV can be computed exactly (see e.g. \cite{Condat2013}), in the case of $\TV_{1}$ another (stochastic) notion of inexactness of the proximal operator can be used, by picking an exact proximal point with respect to the one-dimensional total variation of a randomly drawn single row or column. The use of such stochastic approximations of the proximal operators in forward-backward type Langevin has been investigated in \cite{Salim2019}. The resulting convergence guarantees are comparable to ours, with the main difference that one has to bound the second moment of the stochastic estimator of the proximal points instead of the inexactness level $\epsilon$ used here. While the setting of \cite{Salim2019} is computationally especially suited for cases where $G$ has a splitting as a separable sum with many terms, it only covers distributions that have a positive Lebesgue density everywhere (i.e. $G(x) < \infty$ for all $x \in \bbR^d$).
\end{remark}
%
%
%
%
%
%
%
\subsection{Nonasymptotic and Asymptotic Convergence}
We are ready to state our convergence results for iPGLA. Note that the bounds generalize the ones of \cite{Salim2020} for flexible step sizes and inexactness levels $\epsilon_k$ of the approximate evaluation of the proximal mappings. For iterates $X^k$ with law $\mu^k$ and step sizes $\gamma_k$ as in \cref{algo:iPGLA}, we will refer to the following descent condition 
\begin{equation}\label{eq:descent_condition}
    \bbE_{X\sim \mu^k} \left[ F(X-\gamma_k\nabla F(X)) - F(X) + \frac{\gamma_k}{2} \norm{\nabla F(X)}^2 \right] \le 0,
\end{equation}
which is ensured, e.g., by choosing $\gamma_k \le L^{-1}$.

\begin{theorem}\label{thm:nonasymptotic_result_type2error}
Let \cref{assumption1,assumption2,assumption3} be satisfied and let $(X^k)_k$ be generated by \cref{algo:iPGLA}. If the step size $\gamma_k$ satisfies \eqref{eq:descent_condition}, then for any $\theta \in [0,1)$ it holds
\begin{align*}
    \calW_2^2(\mu^{k+1},\mu^\ast) \le (1-\lambda_F \gamma_k) &\calW_2^2(\mu^k, \mu^\ast) - \gamma_k (\theta \lambda_{G^\ast} + \gamma_k) \calW_2^2(\nu^{k+1},\nu^\ast) \\
    &- 2\gamma_k \left( \calL(\mu^{k+2/3}, \psi^\ast) - \calL(\mu^\ast, \psi^{k+1}) \right) + \gamma_{k}^2\tilde{C} +  \frac{2\gamma_k \epsilon_k}{1-\theta},
\end{align*}
where $\tilde{C} = 2Ld + \int \norm{\nabla G}^2 \rmd \mu^\ast(x) < \infty$.
\end{theorem}

In order to prove the theorem, we need some auxiliary results.

\begin{lemma}\label{lem:wasserstein_error_inexact_prox_step}
Let \cref{assumption1,assumption2,assumption3} be satisfied. For any $\gamma_k > 0$, $\theta \in [0,1)$ and any $\mu \in \calP_2(\calX)$ and $\mu$-measurable function $\psi:\calX \to \calX$ it holds
\begin{equation*}
\begin{split}
    \calW_2^2(\mu^{k+1}, \mu) \le \calW_2^2(\mu^{k+2/3}, \mu) - \gamma_k (\theta\lambda_{G^\ast} + \gamma_k) \calW_2^2(\nu^{k+1},\psi\#\mu) + \gamma_k^2 \bbE_{Y\sim \psi\#\mu} \left[\norm{Y}^2 \right] \\
    - 2\gamma_k \left( \bbE \left[ \inpro{X^{k+2/3}}{Y} \right] - \bbE \left[ \inpro{X}{Y^{k+1}}\right] + \calE_{G^\ast}(\nu^{k+1}) - \calE_{G^\ast}(\psi\#\mu) \right) + 2\gamma_k \frac{\epsilon_{k}}{1-\theta},
\end{split}
\end{equation*}
where the expectations in the second line are over the variables with the joint distributions $(X^{k+2/3},Y) \sim  (T_{\mu\to\mu^{k+2/3}},\psi)\#\mu$ and $(X,Y^{k+1}) \sim (I, (I - S_k)/{\gamma_k} \circ T_{\mu \to \mu^{k+2/3}})\#\mu $, respectively.
\end{lemma}

\begin{proof}
The proof is a modification from Lemma 4 in \cite{Salim2020}, where instead of the standard convexity inequality for subgradients, we use  \cref{lem:strong_convexity_epsilon_subdifferential} for inexact subgradients.
Let $\bar{x} \in \calX$ be fixed, $\hat{x}$ the $\epsilon_k$-inexact proximal point $\hat{x} = S_k(\bar{x}) \approx^{\epsilon_k} \prox_{\gamma_k G}(\bar{x})$. Using \cref{lem:convex_duality_epsilon_subdifferential} to rewrite the condition for the inexact proximal mapping gives
$$ \hat{y} := \frac{\bar{x}-\hat{x}}{\gamma_k} \in \partial_{\epsilon_k} G(\hat{x}) \ \Leftrightarrow\ \hat{x} \in \partial_{\epsilon_k} G^\ast(\hat{y}). $$
Since $\calX$ is a Hilbert space, for any $x\in\calX$ it holds
\begin{equation}\label{inexactProxLemmaNormDecomposition}
    \norm{\hat{x}-x}^2 = \norm{\bar{x}-x}^2 - \norm{\hat{x}-\bar{x}}^2 + 2 \inpro{\hat{x}-\bar{x}}{\hat{x}-x}.
\end{equation}
In order to reformulate the last term, we can apply \cref{lem:strong_convexity_epsilon_subdifferential} using the (strong) convexity of $G^\ast$ so that for any $x,y \in \calX$ and any $\theta \in [0,1)$ we obtain
\begin{align*}
    \inpro{\hat{x}-\bar{x}}{\hat{x}-x} &= \gamma_k \inpro{\hat{y}}{x} - \gamma_k \inpro{\hat{y}}{\hat{x}} \\
    &= \gamma_k \inpro{\hat{y}}{x} - \gamma_k \inpro{\hat{x}}{y} + \gamma_k \inpro{\hat{x}}{y-\hat{y}}\\
    &\overset{\text{}}{}\le \gamma_k \inpro{\hat{y}}{x} - \gamma_k \inpro{\hat{x}}{y} + \gamma_k \left( G^\ast(y) - G^\ast(\hat{y}) - \frac{\theta \lambda_{G^\ast}}{2} \norm{\hat{y}- y}^2 + \frac{\epsilon_{k}}{1-\theta} \right).
\end{align*}
Plugging this into \eqref{inexactProxLemmaNormDecomposition} gives 
\begin{align*}
    \norm{\hat{x}-x}^2 &\le \norm{\bar{x}-x}^2 - \norm{\hat{x}-\bar{x}}^2 \\
    &\qquad + 2 \gamma_k \left( \inpro{\hat{y}}{x} - \inpro{\hat{x}}{y} + G^\ast(y) - G^\ast(\hat{y}) - \frac{\theta \lambda_{G^\ast}}{2} \norm{\hat{y}- y}^2 + \frac{\epsilon_{k}}{1-\theta} \right)\\
    &= \norm{\bar{x}-x}^2 - \gamma_k^2 \norm{\hat{y}}^2 \\
    &\qquad +2\gamma_k \left( \inpro{\hat{y}}{x} - \inpro{\bar{x}}{y} + \gamma_k \inpro{\hat{y}}{y} + G^\ast(y) - G^\ast(\hat{y}) - \frac{\theta\lambda_{G^\ast}}{2} \norm{\hat{y}- y}^2 + \frac{\epsilon_{k}}{1-\theta} \right)\\
    &= \norm{\bar{x}-x}^2 - \gamma_k(\theta \lambda_{G^\ast} + \gamma_k) \norm{\hat{y}-y}^2 + \gamma_k^2 \norm{y}^2 \\
    & \qquad + 2\gamma_k \left( \inpro{\hat{y}}{x} - \inpro{\bar{x}}{y} + G^\ast(y) - G^\ast(\hat{y}) + \frac{\epsilon_{k}}{1-\theta} \right),
\end{align*}
where we used the definition of $\hat y$ in the first equality. 
Assume now that the variables $x,\bar x,y$ are marginally distributed with laws $x\sim\mu$, $ \bar{x} \sim \mu^{k+2/3} $ and $y \sim \nu = \psi\#\mu$ respectively. This implies that $\hat{x} \sim S_k \# \mu^{k+2/3} = \mu^{k+1}$ and $\hat{y} \sim \frac{I-S_k}{\gamma_k} \# \mu^{k+2/3} = \nu^{k+1} $. With the definition of the Wasserstein distance, the last inequality implies
\begin{align*}
    \calW_2^2(\mu^{k+1},\mu) &\le \bbE\left[\norm{\bar{x}-x}^2\right] - \gamma_k(\theta \lambda_{G^\ast} + \gamma_k) \calW_2^2(\nu^{k+1},\nu) + \gamma_k^2 \bbE\left[\norm{y}^2\right]\\
    &\qquad + 2\gamma_k \left( \bbE\left[\inpro{\hat{y}}{x}\right] - \bbE\left[\inpro{\bar{x}}{y}\right] + \calE_{G^\ast}(\nu) - \calE_{G^\ast}(\nu^{k+1}) + \frac{\epsilon_{k}}{1-\theta} \right),
\end{align*}
Since we only assumed the marginal distributions of $x,y,\bar x$, this holds for any coupling of these variables. In particular we can consider $x,\bar x$ to be optimally coupled in the sense that their joint distribution realizes the Wasserstein distance, i.e. $\calW_2^2(\mu^{k+2/3}, \mu) = \bbE[\norm{\bar{x}-x}^2]$. The coupling can be written in terms of a transport map \cite{Villani2003} which we denote by $T_{\mu\to\mu^{k+2/3}}$ so that $(x,\bar x) \sim (I,T_{\mu\to\mu^{k+2/3}}) \# \mu$. The remaining variables are then coupled to $x$ by $\hat x = (S_k \circ T_{\mu\to\mu^{k+2/3}})(x) = $ and $\hat{y} = (\frac{I-S_k}{\gamma_k}\circ T_{\mu\to\mu^{k+2/3}})(x) = \psi^{k+1}(x)$. Letting further $y = \psi(x)$, the last inequality for this particular coupling of $x,\bar x, y$ is just the desired result.
\end{proof}

We recall three Lemmata from the analysis in \cite{Durmus2019a}, their respective Lemmata 3, 5 and 28. The first bounds the development of the entropy along the solutions.
\begin{lemma}\label{lem:wasserstein_error_entropy_step}
For any $\gamma_k > 0$, the entropy values of the iterates satisfy
$$ \calH(\mu^{k+2/3}) - \calH(\mu^\ast) \le \frac{1}{2\gamma_k}\left( \calW_2^2(\mu^{k+1/3}, \mu^\ast) - \calW_2^2(\mu^{k+2/3},\mu^\ast) \right). $$
\end{lemma}

The next Lemma bounds the error in the smooth part $F$ of the potential by the entropy gradient flow step.
\begin{lemma}\label{lem:potential_error_F_entropy_step}
For any $\gamma_k > 0$, we have
$$ \calE_F(\mu^{k+2/3}) - \calE_F(\mu^{k+1/3}) \le L d \gamma_k. $$
\end{lemma}

The last auxiliary result for \cref{thm:nonasymptotic_result_type2error} quantifies the decrease in potential energy along under a descent condition. Note that this formulation is slightly more general than the typical assumption $\gamma \le L^{-1}$ in \cite{Durmus2019a}.
\begin{lemma}\label{lem:wasserstein_error_F_exact_gradient_step_flexible_stepsize}
If $\gamma_k > 0$ satisfies the descent condition \eqref{eq:descent_condition}, then it holds
$$ 2\gamma_k (\calE_F(\mu^{k+1/3}) - \calE_F(\mu^\ast)) \le (1 - \lambda_F \gamma_k) \calW_2^2(\mu^k, \mu^\ast) - \calW_2^2(\mu^{k+1/3}, \mu^\ast). $$
\end{lemma}
\begin{proof} The proof is a straightforward adaption of the proof of Lemma 28 in \cite{Durmus2019a}. Let $x,\tilde{x} \in \calX$ arbitrary. Using strong convexity, one obtains
\begin{align*}
    2\gamma_k &\left( F(x-\gamma_k \nabla F (x)) - F(y) \right) \\
    &= 2\gamma_k \left( F(x-\gamma_k \nabla F (x)) - F(x) \right) + 2\gamma_k (F(x) - F(y))\\
    &\le 2\gamma_k \left( F(x-\gamma_k \nabla F (x)) - F(x) \right) + 2\gamma_k \left( \inpro{\nabla F(x)}{x-y} - \frac{\lambda_F}{2} \norm{x-y}^2 \right)\\
    &= (1- \lambda_F \gamma_k) \norm{x-y}^2 - \norm{x-\gamma_k\nabla F(x) - y}^2 + \gamma_k^2 \norm{\nabla F(x)}^2 \\
    &\qquad + 2\gamma_k \left( F(x-\gamma_k \nabla F (x)) - F(x) \right).
\end{align*}
We now let $y \sim \mu^\ast$ and $x\sim \mu^k$, which also implies $x-\gamma_k \nabla F(x) \sim \mu^{k+1/3}$. Adding $\norm{x-\gamma_k\nabla F(x) - y}^2$ on both sides and taking expectations, we obtain
\begin{align*}
    2\gamma_k &\left( \calE_F(\mu^{k+1/3}) - \calE_F(\mu^\ast) \right) + \calW_2^2(\mu^{k+1/3}, \mu^\ast) \\
    &\le \bbE \left[ 2\gamma_k \left( F(x-\gamma_k \nabla F (x)) - F(y) \right) \right] + \bbE \left[ \norm{x-\gamma_k\nabla F(x) - y}^2 \right] \\
    &\le (1- \lambda_F \gamma_k) \bbE\left[ \norm{x-y}^2 \right] + \bbE\left[\gamma_k^2 \norm{\nabla F(x)}^2 + 2\gamma_k (F(x-\gamma_k\nabla F(x)) - F(x)) \right]\\
    &\le (1- \lambda_F \gamma_k) \bbE\left[ \norm{x-y}^2 \right].
\end{align*}
Since the last inequality holds for any couplings of $x\sim \mu^k$ and $y \sim \mu^\ast$, we can take the infimum over all couplings and obtain the desired result.
\end{proof}

We are now able to prove the theorem.
\begin{proof}[Proof of \cref{thm:nonasymptotic_result_type2error}]
Firstly, \cref{lem:wasserstein_error_inexact_prox_step} with the choice $\mu = \mu^\ast$, $\psi = \psi^\ast = \nabla G$ gives
\begin{equation*}
\begin{split}
    \calW(\mu^{k+1}, \mu^\ast) \le \calW_2^2(\mu^{k+2/3}, \mu^\ast) - \gamma_k (\theta\lambda_{G^\ast} + \gamma_k) \calW_2^2(\nu^{k+1},\nu^\ast) + \gamma_k^2 \bbE_{Y\sim \nu^\ast} \left[\norm{Y}^2 \right] \\
    - 2\gamma_k \left( \bbE \left[ \inpro{X^{k+2/3}}{Y^\ast} \right] - \bbE \left[ \inpro{X^\ast}{Y^{k+1}}\right] + \calE_{G^\ast}(\nu^{k+1}) - \calE_{G^\ast}(\nu^\ast) \right) + 2\gamma_k \frac{\epsilon_{k+1}}{1-\theta}
\end{split}
\end{equation*}
with expectations over $(X^{k+2/3},Y^\ast) \sim  (T_{\mu^\ast\to\mu^{k+2/3}},\nabla G)\#\mu^\ast$ and $(X^\ast,Y^{k+1}) \sim (I, \psi^{k+1})\#\mu^\ast $ in the second line.
Next, \cref{lem:wasserstein_error_entropy_step,lem:potential_error_F_entropy_step} give
$$ 2\gamma_k \left(\calH(\mu^{k+2/3}) - \calH(\mu^\ast)\right) \le  \calW_2^2(\mu^{k+1/3}, \mu^\ast) - \calW_2^2(\mu^{k+2/3},\mu^\ast) $$
and
$$ 2\gamma_k \left(\calE_F(\mu^{k+2/3}) - \calE_F(\mu^{k+1/3})\right) \le 2 L d \gamma_k^2. $$
Under the assumption that the descent condition is satisfied we can use \cref{lem:wasserstein_error_F_exact_gradient_step_flexible_stepsize} to get
$$ 2\gamma_k (\calE_F(\mu^{k+1/3}) - \calE_F(\mu^\ast)) \le (1 - \lambda_F \gamma_k) \calW_2^2(\mu^k, \mu^\ast) - \calW_2^2(\mu^{k+1/3}, \mu^\ast). $$
The statement of the theorem is the sum of these four inequalities.
\end{proof}

We now want to iterate the nonasymptotic result in \cref{thm:nonasymptotic_result_type2error} to get a decay rate in Wasserstein distance with a closed form of the bias.

\begin{theorem}[Fixed step size]\label{thm:asymptotic_result_type2error_fixedSteps}
Let \cref{assumption1,assumption2,assumption3} be satisfied and $X^k$ be generated by \cref{algo:iPGLA} with $\gamma_k = \gamma \le L^{-1}$ for all $k$. Then for any $\theta \in [0,1)$, the following statements hold:
\begin{enumerate}[label=(\roman*)]
    \item\label{item:asymptotic_result_fixedSteps_type2error_convexOnly} If $\epsilon_k \le \epsilon$ for all $k$, then 
        \begin{equation}\label{eq:type2_asymptotic_bound_dual_wasserstein}
            \min_{1\le k \le K} \calW_2^2(\nu^k,\nu^\ast) \le \frac{1}{\gamma (\theta \lambda_{G^\ast} + \gamma)K} \calW_2^2(\mu^0,\mu^\ast) + C_{\theta}(\gamma,\epsilon)
        \end{equation}
        where $C_{\theta}(\gamma,\epsilon) = \frac{\tilde{C}\gamma(1-\theta) + 2\epsilon}{(1-\theta)(\theta \lambda_{G^\ast}+\gamma)}$ and
        \begin{equation}\label{eq:type2_asymptotic_bound_duality_gap}
            \min_{1\le k \le K} \calD(\mu^{k-1/3},\psi^k) \le \frac{1}{2 \gamma K} \calW_2^2(\mu^0,\mu^\ast) + \frac{\tilde{C}}{2}\gamma + \epsilon.
        \end{equation}
        with $\tilde{C}$ as in \cref{thm:nonasymptotic_result_type2error}. If further $\lambda_F > 0$, then
        \begin{equation}\label{eq:type2_asymptotic_bound_primal_wasserstein_boundedErrors}
            \calW_2^2(\mu^{K},\mu^\ast) \le (1-\lambda_F \gamma)^K \calW_2^2(\mu^0, \mu^\ast) + \frac{\gamma}{\lambda_F} \tilde{C} +  \frac{2\epsilon}{\lambda_F}.
        \end{equation}
    \item\label{item:asymptotic_result_fixedSteps_type2error_stronglyconvex_monotonousErrors} If $F$ is $\lambda_F$-strongly convex and $(\epsilon_k)_k$ is a monotonically decreasing sequence, then
        \begin{equation}\label{eq:type2_asymptotic_bound_primal_wasserstein_decreasingErrors}
            \calW_2^2(\mu^{K},\mu^\ast) \le (1-\lambda_F \gamma)^K \calW_2^2(\mu^0, \mu^\ast) + \frac{\gamma}{\lambda_F} \tilde{C} +  \frac{2}{K\lambda_F} \sum_{k=0}^{K-1} \epsilon_k.
        \end{equation}
\end{enumerate}
\end{theorem}

We give some remarks on the interpretation of these results.
\begin{remark}
    If all parameters on the right hand side of \eqref{eq:type2_asymptotic_bound_dual_wasserstein} are known (or estimated), one can optimize the parameter $\theta$ to arrive at an optimal bound and the number of iterations necessary to achieve a certain level of accuracy in the dual variable. If $\lambda_{G^\ast} = 0$, then the optimal $\theta$ is zero and $\inf_\theta C_{\theta}(\gamma,\epsilon) = \tilde{C} + 2\epsilon/\gamma$. If $\lambda_{G^\ast} > 0$ then for any $\theta > 0$ we have $C_\theta(\gamma,\epsilon) \to 0$ as $\gamma, \epsilon \to 0$.

    Note further that in part \ref{item:asymptotic_result_fixedSteps_type2error_stronglyconvex_monotonousErrors}, the bias induced by the errors $\epsilon_k$ vanishes as $K\to \infty$ if and only if $(\epsilon_k)_k$ converges to 0. In this case, independently of the decay rate of $\epsilon_k$, the bias scales linearly with $\gamma$ in the limit as is the case in the error-free version of the algorithm \cite{Salim2020}.
\end{remark}

\begin{remark}
    The results show that with decreasing error $\epsilon$ in the proximal map, the approximation to the target $\mu^\ast$ gets better. If in practice hyperparameters for $\gamma$, $\epsilon$ and $K$ can be set in advance, this leads to an implicit optimization problem of finding the right parameters since higher accuracy in the proximal points usually comes with increasing computational cost at each iteration. Consider for example that we are in the setting of \eqref{eq:type2_asymptotic_bound_primal_wasserstein_boundedErrors} and we want to set optimal parameters such that $\calW_2^2(\mu^K,\mu^\ast) \le \delta$ for some bound $\delta > 0$. Assuming that the cost of each iteration is dominated by the accuracy $\epsilon$ indicates that there is an optimal $\epsilon$: Making large steps and picking $\epsilon$ as small as possible is suboptimal since then the cost per iteration becomes prohibitively large. On the other hand, choosing $\epsilon$ as large as possible close to $\frac{\lambda_F \delta}{2}$ in order to save on cost per iteration means $\gamma$ has to be small, implying $1-\lambda_F \gamma \approx 1$ so that the necessary number of iterations $K$ becomes large.
\end{remark}

\begin{proof}[Proof of \cref{thm:asymptotic_result_type2error_fixedSteps}]
For part \ref{item:asymptotic_result_fixedSteps_type2error_convexOnly}, we assume only that the errors are bounded by some upper value $\epsilon$. \cref{thm:nonasymptotic_result_type2error} gives
\begin{equation*}
    \gamma (\theta \lambda_{G^\ast} + \gamma) \calW_2^2(\nu^{k+1},\nu^\ast) + 2\gamma \calD(\mu^{k+2/3},\psi^{k+1}) 
    \le \calW_2^2(\mu^k, \mu^\ast) - \calW_2^2(\mu^{k+1},\mu^\ast) + \gamma^2\tilde{C} +  \frac{2\gamma \epsilon}{1-\theta}.
\end{equation*}
Summing over $k=0,\dots,K-1$ gives
\begin{equation*}
    \sum_{k=1}^{K}\gamma (\theta \lambda_{G^\ast} + \gamma) \calW_2^2(\nu^k,\nu^\ast) + 2\gamma\sum_{k=1}^K \calD(\mu^{k-1/3},\psi^k) 
    \le \calW_2^2(\mu^0, \mu^\ast) - \calW_2^2(\mu^K,\mu^\ast) + K \gamma^2\tilde{C} + \frac{2 K \gamma \epsilon}{1-\theta}.
\end{equation*}
This implies
\begin{equation*}
    \gamma (\theta \lambda_{G^\ast} + \gamma) \min_{k=1,\dots,K} \calW_2^2(\nu^j,\nu^\ast) + 2\gamma \min_{k=1,\dots,K} \calD(\mu^{k-1/3},\psi^k) \le \frac{1}{K} \calW_2^2(\mu^0, \mu^\ast) + \gamma^2\tilde{C} + \frac{2 \gamma \epsilon}{1-\theta},
\end{equation*}
which in turn directly implies both \eqref{eq:type2_asymptotic_bound_dual_wasserstein} and, for $\theta=0$, \eqref{eq:type2_asymptotic_bound_duality_gap}. If $F$ is $\lambda_F$-strongly convex and $\gamma_k = \gamma > 0$ fixed then \cref{thm:nonasymptotic_result_type2error} with $\theta = 0$ implies
$$ \calW_2^2(\mu^{k+1},\mu^\ast) \le (1-\lambda_F \gamma) \calW_2^2(\mu^k, \mu^\ast) + \gamma^2\tilde{C} + 2\gamma \epsilon_{k} $$
By iterating, we get
\begin{equation}\label{eqn:inequalityproof_asymptoticbias_strongconvexcase}
    \calW_2^2(\mu^{K},\mu^\ast) \le (1-\lambda_F \gamma)^K \calW_2^2(\mu^0, \mu^\ast) + \sum_{k=0}^{K-1} (1-\lambda_F \gamma)^k \left(\gamma^2 C + 2 \gamma \epsilon_{K-1-k}\right).
\end{equation}
The last part of \ref{item:asymptotic_result_fixedSteps_type2error_convexOnly} follows from \eqref{eqn:inequalityproof_asymptoticbias_strongconvexcase} by using $\epsilon_k \le \epsilon$ and $$\sum_{k=0}^{K-1} (1-\lambda_F \gamma)^k = \frac{1-(1-\lambda_F \gamma)^K}{\lambda_F \gamma} \le \frac{1}{\lambda_F \gamma}.$$

For \ref{item:asymptotic_result_fixedSteps_type2error_stronglyconvex_monotonousErrors}, if $\epsilon_k$ is assumed to be monotonically decreasing we can use Chebyshev's inequality to bound the right side of \eqref{eqn:inequalityproof_asymptoticbias_strongconvexcase} and obtain
\begin{align*}
    \calW_2^2(\mu^{K},\mu^\ast) &\le (1-\lambda_F \gamma)^K \calW_2^2(\mu^0, \mu^\ast) + \frac{\gamma}{\lambda_F} C + 2 \gamma \sum_{k=0}^{K-1} (1-\lambda_F \gamma)^k \epsilon_{K-1-k}\\
    &\le (1-\lambda_F \gamma)^K \calW_2^2(\mu^0, \mu^\ast) + \frac{\gamma}{\lambda_F} C +  \frac{2}{K\lambda_F} \sum_{k=0}^{K-1} \epsilon_{k}.
\end{align*}
\end{proof}

In the same style as for other stochastic optimization and sampling algorithms, we can decrease the step sizes during the iteration at a certain rate to successively reduce the bias. The following theorem shows that in the case of strongly convex $F$ and with the right choice of step sizes, the only condition for convergence to the target distribution is that the errors $\epsilon_k$ decay monotonically to zero.

\begin{theorem}[Decaying step sizes]\label{thm:asymptotic_result_type2error_decaying_stepsizes}
Let \cref{assumption1,assumption2,assumption3} be satisfied with $\lambda_F > 0$ and assume the step sizes $\gamma_k$ satisfy all of the following:
\begin{enumerate}
    \item $\gamma_0 = L^{-1}$,
    \item $\gamma_{k-1} \ge \gamma_k \ge \frac{\gamma_{k-1}}{1+\lambda_F}$ for all $k \ge 1$ and
    \item $\lim_{k\to\infty} \gamma_K = 0$, the sum of all step sizes $\sum_{k=0}^{K-1} \gamma_k$ diverges to $\infty$ while the sums $$A_K := \sum_{k=0}^{K-1} \gamma_k \prod_{l=k+1}^{K-1} (1-\lambda_F \gamma_l)$$ remain bounded, i.e. $A_K \le M < \infty$ as $K\to \infty$.
\end{enumerate}
Then, if the errors $\epsilon_k$ are monotonically decreasing with $\epsilon_k \to 0$, we obtain $\calW(\mu^K,\mu^\ast)\to 0$ as $K\to \infty$.
\end{theorem}

\begin{proof}[Proof of \cref{thm:asymptotic_result_type2error_decaying_stepsizes}]
Iterating \cref{thm:nonasymptotic_result_type2error}, we obtain
\begin{equation*}
\calW_2^2(\mu^K,\mu^\ast) \le \prod_{k=0}^{K-1} (1-\lambda_F \gamma_k) \calW_2^2(\mu^0, \mu^\ast) + \sum_{k=0}^{K-1} (\gamma_k^2\tilde{C} + 2\gamma_k \epsilon_{k}) \prod_{j=k+1}^{K-1}(1-\lambda_F \gamma_j),
\end{equation*}
with $\tilde{C} = 2Ld + \int \norm{\nabla G}^2 \rmd \mu^\ast(x) < \infty$. We rewrite this denoting $a_k := \gamma_k \prod_{j=k+1}^{K-1}(1-\lambda_F \gamma_j)$ and using the basic inequality $\log(1-x) \le -x$ to obtain
\begin{equation*}
    \calW_2^2(\mu^K,\mu^\ast) \le \exp\left( \sum_{k=0}^{K-1} \log(1-\lambda_F \gamma_k) \right) \calW_2^2(\mu^0, \mu^\ast) + \sum_{k=0}^{K-1} (\gamma_k\tilde{C} + 2 \epsilon_{k}) a_k
\end{equation*}
The sequence $(\tilde{C}\gamma_k + 2\epsilon_{k})$ is monotonically decreasing by assumption. The sequence $a_k$ is monotonically increasing due to the requirement $\gamma_{k+1} \ge \frac{1}{1+\lambda_F} \gamma_k$, as $a_{k}/a_{k+1} = \frac{\gamma_k (1-\lambda_F \gamma_{k+1})}{\gamma_{k+1}} = \frac{\gamma_k}{\gamma_{k+1}} - \lambda_F \ge 1 $. Hence we can use Chebyshev's sum inequality and obtain 
\begin{equation*}
    \calW_2^2(\mu^K,\mu^\ast) \le \exp\left( - \lambda_F \sum_{k=0}^{K-1} \gamma_k \right) \calW_2^2(\mu^0, \mu^\ast) + \frac{1}{K} \left(\sum_{k=0}^{K-1} (\gamma_k\tilde{C} + 2 \epsilon_{k})\right) \left(\sum_{k=0}^{K-1} a_k\right).
\end{equation*}
The first term on the right hand side vanishes to zero as $K \to \infty$ since we assumed that $\sum_{k=0}^{\infty} \gamma_k = \infty$. The second term converges to zero because by assumption $A_K = \sum_{k=0}^{K-1} a_k$ remains bounded as $K\to \infty$ and $\frac{1}{K} \sum_{k=0}^{K-1} (\gamma_k\tilde{C} + 2 \epsilon_{k}) $ converges to 0.
\end{proof}

Since \cref{thm:asymptotic_result_type2error_decaying_stepsizes} poses several constraints on the step size, we give a simple choice of $\gamma_k$ satisfying all conditions in the following remark.

\begin{remark}\label{rem:stepsize_choice_decaying_W2dist_inexact_PLA}
A choice of step sizes meeting the requirements of \cref{thm:asymptotic_result_type2error_decaying_stepsizes} is $\gamma_0 = L^{-1}$ and $\gamma_k = \min\{\gamma_{k-1},\max\{\frac{C'}{k}, \frac{\gamma_{k-1}}{1+\lambda_F}\}\}$, for any constant $C' \ge \lambda_F^{-1}$. See \cref{fig:decreasing_step_size_choice} for a visualization of this choice. \\
To see that this is valid, note that $\gamma_{k-1} \ge \gamma_k \ge \frac{\gamma_{k-1}}{1+\lambda_F}$ is always satisfied, as well as the descent condition since $\gamma_k \le \gamma_0= L^{-1}$. The term $\frac{C'}{k}$ ensures that the step sizes decay to 0 slowly enough so that $\sum_{k=0}^{K-1} \gamma_k $ diverges, since by construction there exists an $N$ such that $\gamma_k = \frac{C'}{k}$ for all $k > N$. The final condition is also satisfied since for $K>N$ it holds
\begin{align*}
    A_K = \sum_{k=0}^{K-1} \gamma_k \prod_{l=k+1}^{K-1} (1-\lambda_F \gamma_l) &= \sum_{k=0}^{N} \gamma_k \prod_{l=k+1}^{K-1} \underset{\le 1}{\underbrace{(1-\lambda_F \gamma_l)}} + \sum_{k=N+1}^{K-1} \gamma_k \prod_{l=k+1}^{K-1} (1-\lambda_F \gamma_l) \\
    &\le M' + \sum_{k=N+1}^{K-1} \frac{C'}{k} \prod_{l=k+1}^{K-1} \frac{l-\lambda_F C'}{l} \\
    &= M' + \frac{C'}{K-1} \sum_{k=N+1}^{K-1} \prod_{l=k+1}^{K-1} \underset{\le 1}{\underbrace{\frac{l-\lambda_F C'}{l-1}}} \le M' + C' =: M < \infty
\end{align*}
where the constant $M' = \sum_{k=0}^N \gamma_k$ depends only on $N$, not on $K$. The last equality follows from shifting the denominators inside the product by one index.
\end{remark}

\begin{figure}[t]
    \centering
    \hfill%
    \includegraphics[width=0.8\linewidth]{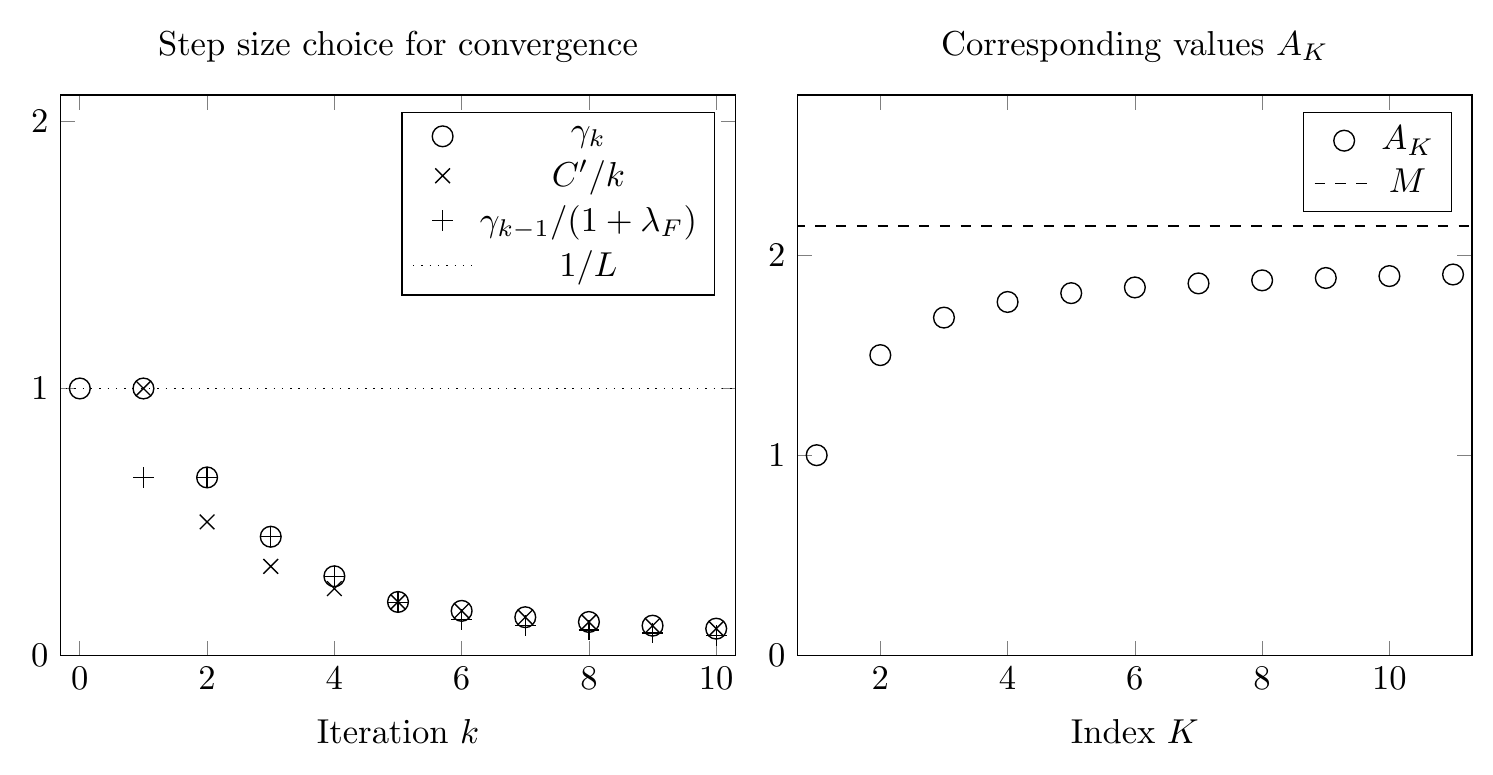}
    \hfill%
    \null%
    \caption{Visualization of the choice of step sizes in \cref{rem:stepsize_choice_decaying_W2dist_inexact_PLA} when $C=1$, $\lambda_F = 0.5$ and $L=1$ (left) and the corresponding values $A_K$ and the upper bound $M=M'+C'$ (right). For $k\ge 5$, it holds $\gamma_k = C'/k$.}
    \label{fig:decreasing_step_size_choice}
\end{figure}

\section{Numerical Results}\label{sec:numerics}
In this section, we validate the theoretical results and apply the sampling algorithm to a range of imaging inverse problems with posterior distributions of the form \eqref{eq:bayes_law}. From the samples, we compute approximations of point estimates like the minimum mean square error (MMSE) and pixelwise standard deviations.

In a first experiment, we validate \cref{thm:asymptotic_result_type2error_fixedSteps} by running the algorithm on a toy example in one dimension in \cref{subsec:1dwasserstein}. For actual images, in \cref{subsec:wavelet_deblurring} we choose a setup in which we know the exact proximal mapping of $G$. This allows us to compare inexact PGLA with its exact special case corresponding to errors $\epsilon=0$ and visualize the additional bias in the stationary distribution due to the errors. In the subsequent experiments we sample from the posterior of typical imaging problems for which the proximal operator of the non-smooth potential $G$ does not have a closed form. We cover an image denoising problem (\cref{subsec:tv_denoising}), an image deblurring problem with Gaussian data and highly ill-conditioned blur kernel (\cref{subsec:tv_deblurring_gauss}) and an image deblurring problem from low count Poisson distributed data (\cref{subsec:tv_deblurring_poisson}). These tests demonstrate the ability of the algorithm to converge to a stationary distribution close to the target, despite the lack of an exactly available proximal mapping. The code for all the experiments is available at \url{https://github.com/lokuger/inexact-proximal-langevin-sampling}.

\subsection{Validation on a Toy Example}\label{subsec:1dwasserstein}
Consider a one dimensional toy example with $F(x) = \frac{1}{2\sigma^2}(x-y)^2$ for all $x\in\bbR$ as the negative log-likelihood of a Gaussian noise distribution for a measurement $y \in \bbR$ and a Laplace distribution as prior with negative log-density $G(x) = \alpha |x|$, $x\in\bbR$. Then the posterior is of the form \eqref{eq:bayes_law} and $F$, $G$ satisfy \cref{assumption1,assumption2,assumption3} with $\lambda_F = L = \sigma^{-2}$ and $\lambda_{G^\ast} = 0$. For this example we can compute that $\tilde C = 1+2\sigma^{-2}$ which allows us to validate the upper bounds from \cref{thm:asymptotic_result_type2error_fixedSteps}

Validating the theoretical bounds like \eqref{eq:type2_asymptotic_bound_primal_wasserstein_boundedErrors} and \eqref{eq:type2_asymptotic_bound_primal_wasserstein_decreasingErrors} requires computing Wasserstein distances between the inaccessible distributions $\mu^k$ and the posterior. We replace the distributions by empirical approximations by generating samples. The distribution $\mu^k$ can be approximated by running several parallel instances of \cref{algo:iPGLA} with independent realizations of the stochastic term and initial values drawn from $\mu^0$. Samples from the posterior $\mu^\ast$ are generated by computing a Markov chain with a time discretized Langevin diffusion step as proposal and an additional Metropolis-Hastings correction step as proposed in \cite{Pereyra2016}. These methods, sometimes also called Metropolis-adjusted Langevin algorithms are typically less efficient due to the additional correction step, but guarantee convergence to the target measure \cite{Roberts1996}.

With the choice $\alpha = \sigma = 1$, we compute $10^5$ parallel Markov chains using \cref{algo:iPGLA} up to $k=10^3$ each and from one single chain another $10^7$ samples with the Metropolis-adjusted algorithm. We then compare the empirical approximations $\tilde \mu^k \approx \mu^k$ with the empirical approximation $\tilde \mu^\ast \approx \mu^\ast$ of the target distribution. In one dimension, the squared Wasserstein distances $\calW_2^2(\tilde \mu^k,\tilde \mu^\ast)$ can be computed efficiently, the computations are carried out using the `Python optimal transport' toolbox \cite{Flamary2021}. The resulting distances and the upper bounds \eqref{eq:type2_asymptotic_bound_primal_wasserstein_boundedErrors} and \eqref{eq:type2_asymptotic_bound_primal_wasserstein_decreasingErrors} are shown in \cref{fig:W2dists_toyexample}. For fixed step sizes, the saturation of the distances at a biased level away from the true target is clearly visible. Running the same experiment with the decaying step size choice from \cref{rem:stepsize_choice_decaying_W2dist_inexact_PLA} and decaying inexactness levels $\epsilon_k$ indicates the convergence to the target measure. The convergence rate depends on the decay rate of $\epsilon_k$.

\begin{figure}[t]
    \centering
    \hfill%
    \begin{subfigure}{0.33\linewidth}%
        \includegraphics[width=\linewidth]{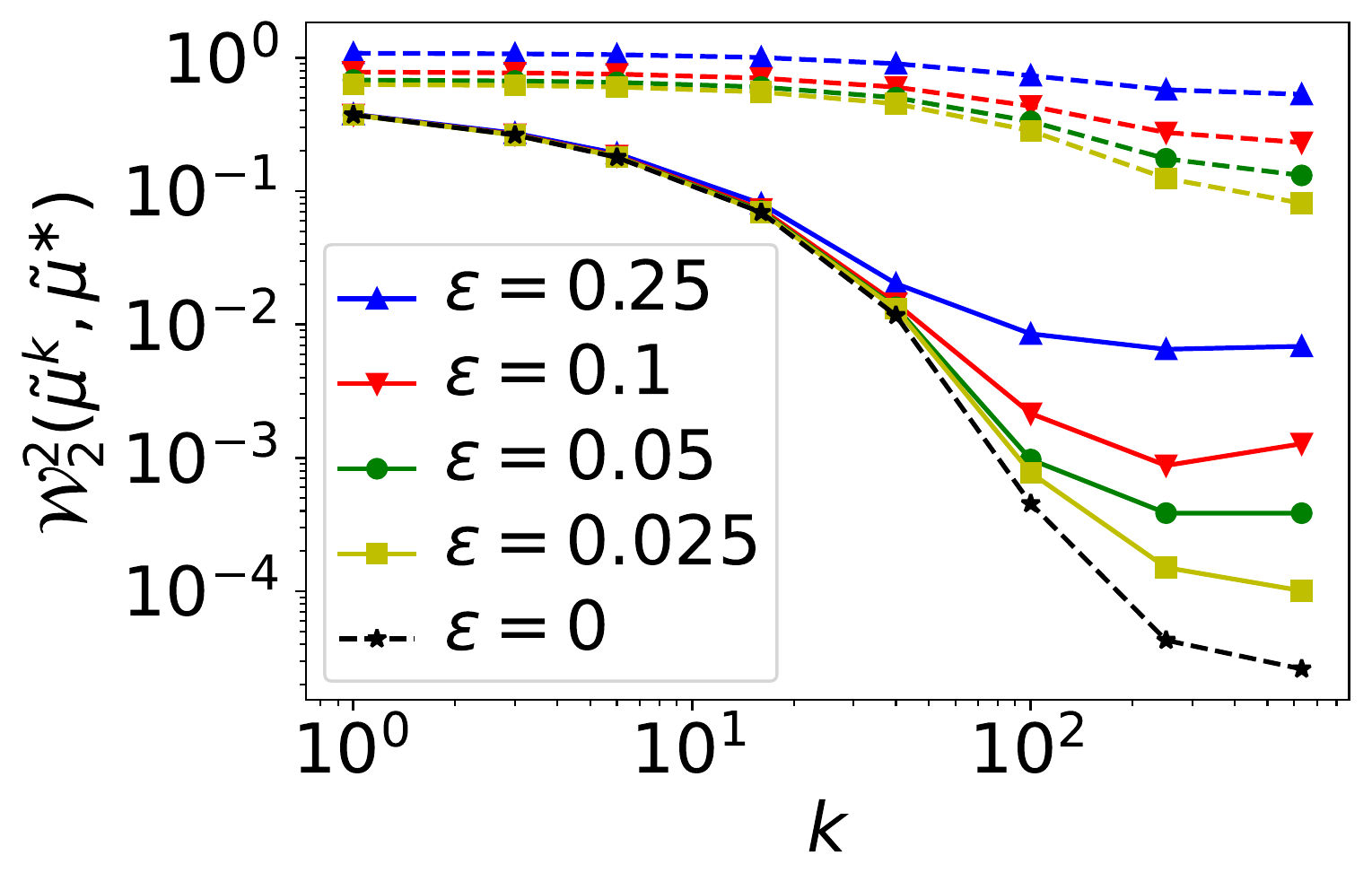}%
        \subcaption{Fixed $\gamma$, fixed $\epsilon$}
    \end{subfigure}%
    \hfill%
    \begin{subfigure}{0.33\linewidth}%
        \includegraphics[width=\linewidth]{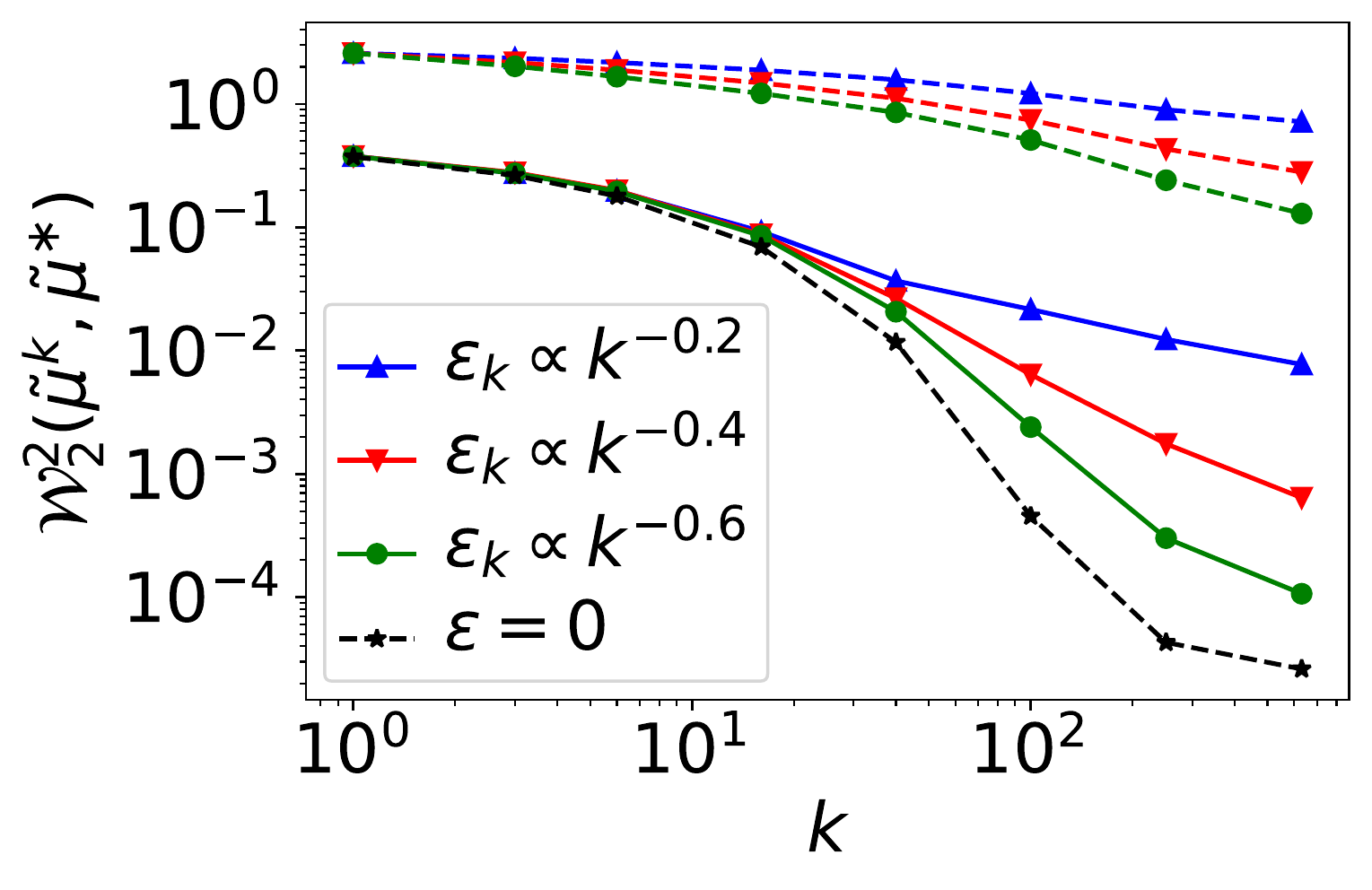}%
        \subcaption{Fixed $\gamma$, decaying $\epsilon_k$}
    \end{subfigure}%
    \hfill%
    \begin{subfigure}{0.33\linewidth}%
        \includegraphics[width=\linewidth]{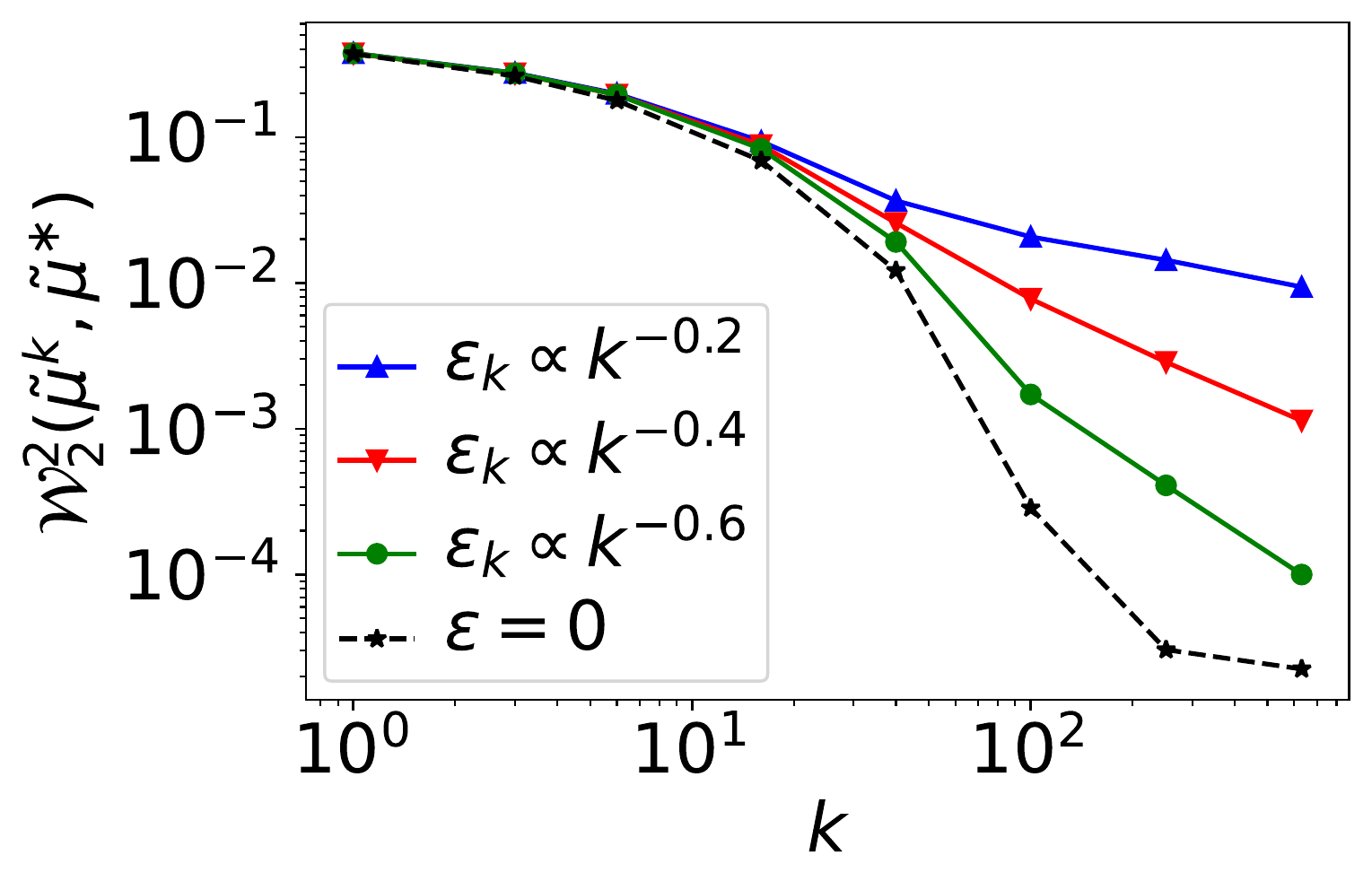}%
        \subcaption{Decaying $\gamma_k$, decaying $\epsilon_k$}
    \end{subfigure}%
    \hfill%
    \null%
    \caption{Squared Wasserstein distances between target and samples generated by \cref{algo:iPGLA} in a 1D toy example. (a): Samples are drawn using a fixed step size $\gamma$ and fixed inexactness level $\epsilon$. The computed squared distances $\calW_2^2(\tilde\mu^k,\tilde\mu^\ast)$ are shown as solid lines, the corresponding upper bound predicted by \eqref{eq:type2_asymptotic_bound_primal_wasserstein_boundedErrors} as dashed lines. 
    (b): Samples are drawn using fixed step size $\gamma$ and inexactness $\epsilon_k \propto k^{-\beta}$ decaying to zero at different rates, with the corresponding upper bounds \eqref{eq:type2_asymptotic_bound_primal_wasserstein_decreasingErrors}.
    (c): Samples are drawn with decaying step sizes $\gamma_k$ chosen as in \cref{rem:stepsize_choice_decaying_W2dist_inexact_PLA} and $\epsilon_k \propto k^{\beta}$ for different rates $\beta<0$. The convergence to the target shown in \cref{thm:asymptotic_result_type2error_decaying_stepsizes} depends in its rate on the decay rate of $\epsilon_k$.}
    \label{fig:W2dists_toyexample}
\end{figure}

\subsection{Wavelet-Based Deblurring -- Comparison of Inexact and Exact PGLA}\label{subsec:wavelet_deblurring}
We consider image deblurring with an $l_1$ regularization term on the image's coefficients in a wavelet basis. The image $x \in \bbR^d$, $d = n_1 n_2$ is blurred by applying a blurring operator $A \in \bbR^{d\times d}$ and then further corrupted by adding normally distributed independent zero-mean noise with variance $\sigma^2$ to each pixel. Implicitly, we represent the image in a wavelet basis as $x = W^Tz$, $z\in \bbR^d$ where $W$ is an orthogonal discrete wavelet transform. The negative log-likelihood is then $F(z) = \norm{AW^Tz-y}_2^2 /2\sigma^2$, a smooth, strongly convex function with strong convexity constant $\lambda_F = \lambda_{\min}(A^\ast A)/\sigma^2$ where $\lambda_{\min}(A^\ast A)$ denotes the smallest eigenvalue of $A^\ast A$. Note that depending on the blur kernel, the strong convexity constant can get very small in the deblurring case when the smallest singular value of $A$ is close to zero.

The $l_1$ regularization on the wavelet coefficients corresponds to assuming a prior of identical centered Laplace distributions for each wavelet coefficient, giving a prior density proportional to $ \exp(-\mu\lVert z \rVert_1)$. The scale parameter $\mu$ takes the role of a regularization parameter on the wavelet coefficients in the denoising experiment. The corresponding energy potential is $G(z) = \mu\lVert z \rVert_1$. The proximal mapping of $G$ is given in closed-form by the soft thresholding operator which for each component $i \in \{1,\dots,d\}$ is given by
$$ (\prox_{\lambda \norm{\cdot}_1}(x))_i = \max(|x_i|-\lambda, 0) \sign(x_i). $$

We compare \cref{algo:iPGLA} for inexact proximal points ($\epsilon > 0$) with its special case of exact proximal points ($\epsilon = 0$) on the blurred and noisy image shown in the middle of the first row of \cref{fig:wavelet_deblurring}. We draw $10^5$ samples from the posterior using \cref{algo:iPGLA} with fixed step size $\gamma = L^{-1}$ and the exact proximal points (corresponding to $\epsilon_k = 0$ for all $k$). We then report the image which is reconstructed from the MMSE estimate wavelet coefficients given by the sample mean.

Afterwards, we run the inexact algorithm with the same step size and same number of samples but different error levels $\epsilon > 0$. This is implemented by deliberately computing the proximal points $\prox_{\gamma G}(X^{k+2/3})$ only approximately. We ensure the approximation as suggested in \cref{subsec:inexact_prox_mappings} by constructing a sequence $z^n$ which converges to the true solution of the dual problem \eqref{eq:prox_problem_dual_form} and stopping the iteration once the duality gap \eqref{eq:duality_gap} is less or equal $\epsilon$. Since we do not know all the constants in the convergence results the accuracy $\epsilon$ is practically chosen depending on the size of the duality gap. Similar to \cite{Rasch2020}, we scale $\epsilon$ w.r.t. a reference constant of the duality gap evaluated for the first sample with $z=0$ and the corresponding primal iterate
$$\epsilon = C_0 \tilde \epsilon,\quad C_0 := \calG(X^{0+2/3},0), $$
and test different values of $\tilde \epsilon \le 1$. Note that the extreme inexact case of taking $\tilde \epsilon=1$ and accepting as inexact dual solution $z=0$ coincides with primal inexact proximal points $X^{k+1} = X^{k+2/3}$ and hence with applying ULA \eqref{eq:ULA} only to the likelihood term and ignoring the $l_1$ regularizing prior. The smaller $\tilde \epsilon$ is chosen, the more exact the approximation to the proximal point. We run the algorithm for $\tilde \epsilon = 10^{-0.1}, 10^{-0.5}, 10^{-2}$ and report the resulting images reconstructed from the MMSE estimates in the second row of \cref{fig:wavelet_deblurring}. Both visually and in PSNR values between the reconstruction and the ground truth we can see the effect of the inexactness in the proximal operator and the resulting regularizing effect of more accurate proximal points.

\begin{figure}[t]
    \centering
    \setlength{\fboxsep}{0pt}%
    \setlength{\fboxrule}{1pt}%
    \newlength{\imageheight}%
    \setlength{\imageheight}{0.28\linewidth}%
    \newlength{\sfwidth}%
    \setlength{\sfwidth}{0.30\linewidth}%
    \hfill%
    \begin{subfigure}{\sfwidth}%
        \centering%
        \includegraphics[height=\imageheight]{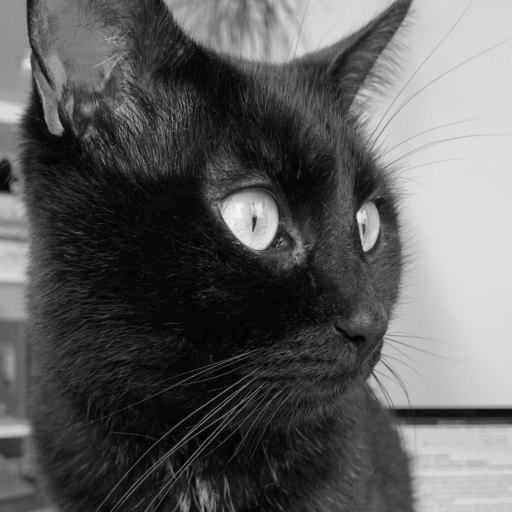}%
        \subcaption*{Ground truth}%
    \end{subfigure}%
    \hfill%
    \begin{subfigure}{\sfwidth}%
        \centering%
        \includegraphics[height=\imageheight]{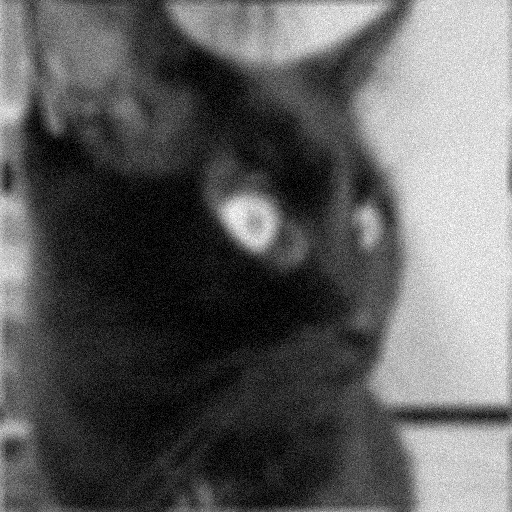}%
        \subcaption*{Blurred and noisy}%
    \end{subfigure}%
    \hfill%
    \begin{subfigure}{\sfwidth}%
        \centering%
        \includegraphics[height=\imageheight]{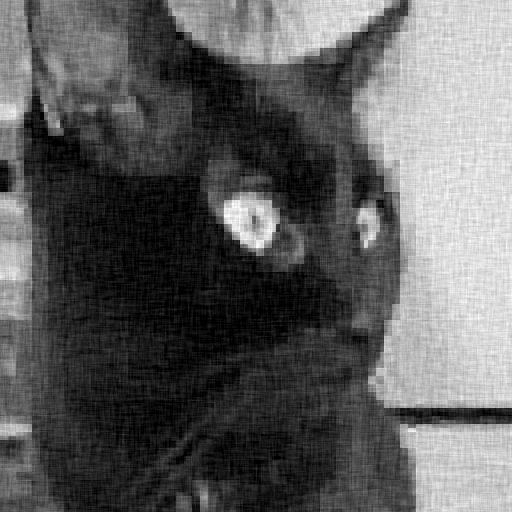}%
        \subcaption*{MMSE exact proximal map, $\textrm{PSNR}=24.40$}%
    \end{subfigure}%
    \hfill\\%
    \hfill%
    \begin{subfigure}{\sfwidth}%
        \centering%
        \includegraphics[height=\imageheight]{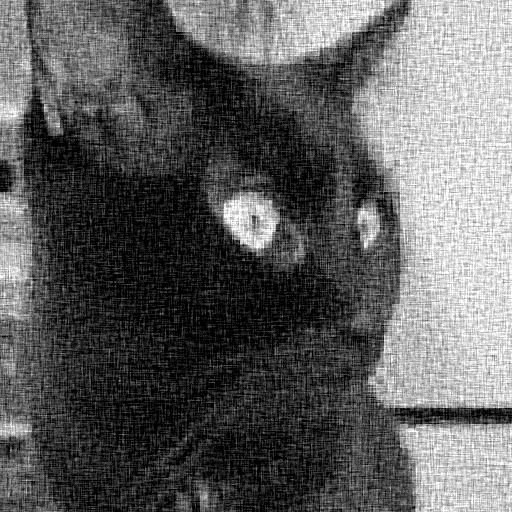}%
        \subcaption*{MMSE inexact proximal map, $\tilde \epsilon = 10^{-0.1}$, $\textrm{PSNR}=19.28$}%
    \end{subfigure}%
    \hfill%
    \begin{subfigure}{\sfwidth}%
        \centering%
        \includegraphics[height=\imageheight]{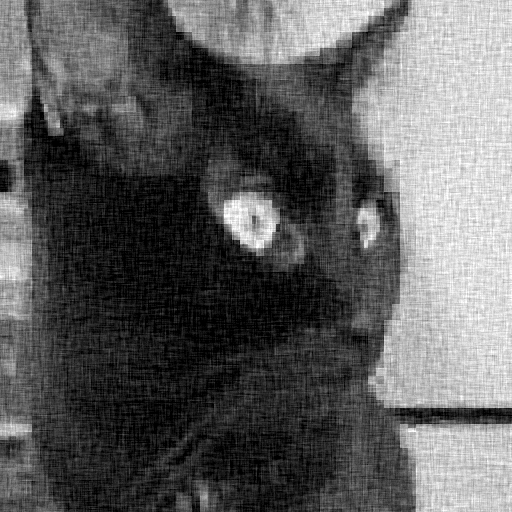}%
        \subcaption*{MMSE inexact proximal map, $\tilde \epsilon = 10^{-0.5}$, $\textrm{PSNR}=22.83$}%
    \end{subfigure}%
    \hfill%
    \begin{subfigure}{\sfwidth}%
        \centering%
        \includegraphics[height=\imageheight]{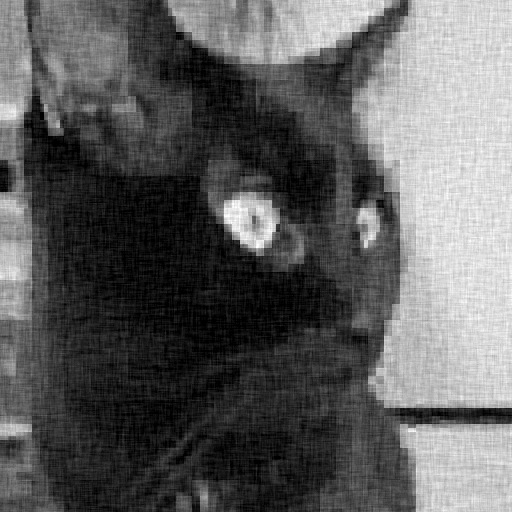}%
        \subcaption*{MMSE inexact proximal map, $\tilde \epsilon = 10^{-2.0}$, $\textrm{PSNR}=24.37$}%
    \end{subfigure}%
    \hfill\\%
    \caption{Comparison of exact and inexact version of PGLA for wavelet-based deblurring. For the inexact version we use errors $\epsilon = C_0 \tilde \epsilon$.}
    \label{fig:wavelet_deblurring}
\end{figure}

\subsection{TV-Regularized Denoising}\label{subsec:tv_denoising}
Next, we consider image denoising and deblurring with a total variation (TV) prior. The negative log-likelihood is given by $F(x) = \norm{x-y}_2^2 /2\sigma^2$. As regularization, we choose a Gibbs prior with negative log-density proportional to $G(x) = \mu \TV(x)$. The inexact proximal points are computed as described in \cref{ex:epsilon_prox_computation_for_TV}.

For denoising, we scale an image of size $512 \times 512$ to gray scale values in $[0,1]$ and set $\sigma = 0.2$.
In \cref{fig:TV_denoising_wheel_referenceimages} we show the ground truth image, corrupted version and MMSE estimate $\hat{X}_{MMSE}$. The MMSE is estimated by drawing $10^5$ samples from the posterior using \cref{algo:iPGLA} with $\gamma=L^{-1}$. Due to the inherent bias of the algorithm, the MMSE estimate may be inaccurate, however here we are mostly interested in the additional error brought in by inexact evaluations of the proximal mapping. As before, we set $\epsilon = C_0 \tilde \epsilon$ with $C_0 = \calG(X^{0+2/3},0)$. The MMSE estimate $\hat{X}_{MMSE}$ in \cref{fig:TV_denoising_wheel_referenceimages} is computed using the small value $\tilde \epsilon = 10^{-5}$. We run the algorithm using different values of $\tilde \epsilon$ between $10^{-2}$ and $10^{-4}$ and report the convergence behaviour of the chains' first moment in dependence of the inexactness level. In particular, we let each chain run until a stopping criterion is satisfied at iteration $k^*$, defined by
$$ k^*(\delta) := \min\left\{ k: \frac{\norm{\bar{X}_k - \hat{X}_{MMSE}}_2}{\norm{\hat{X}_{MMSE}}_2} \le \delta \right\}, $$
where $\bar{X}_k$ is the running mean of the chain at iteration $k$ after a fixed burnin phase. We report the values of $k^\ast$ for several values of $\delta$ and $\tilde \epsilon$ in \cref{fig:convergence_MMSE_TV_denoising}. Naturally, the number of samples (outer iterations) required to reach a certain MMSE accuracy becomes smaller when the proximal accuracy $\tilde \epsilon$ is small. However, computationally it can make sense to choose a larger $\tilde \epsilon$ and solve the inner optimization problem only with a few steps, since the total number of inner iterations (and hence the total CPU time of the chain) may still be smaller, as is depicted in the second plot of the figure.

\begin{figure}[t]
    \centering
    \setlength{\fboxsep}{0pt}%
    \setlength{\fboxrule}{1pt}%
    \setlength{\imageheight}{4.5cm}%
    \newlength{\detailwidth}%
    \setlength{\detailwidth}{1.8cm}%
    \setbox1=\hbox{\includegraphics[height=\imageheight]{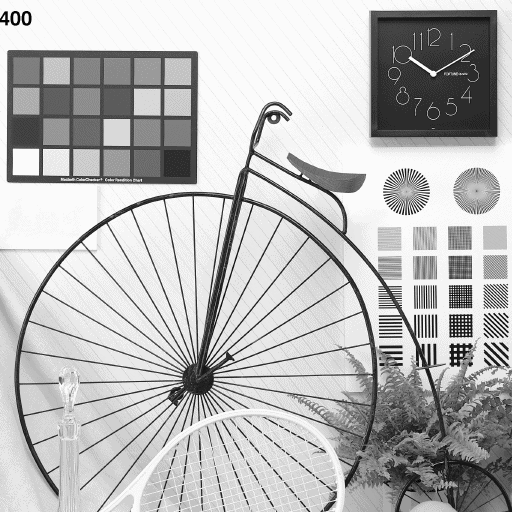}}%
    \hfill%
    \begin{subfigure}{0.3\linewidth}%
        \includegraphics[height=\imageheight]{results/denoise_TV/wheel/ground_truth.png}\llap{\makebox[\wd1][l]{\fbox{\includegraphics[height=\detailwidth]{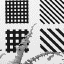}}}}%
        \subcaption*{Ground truth}%
    \end{subfigure}%
    \hfill%
    \begin{subfigure}{0.3\linewidth}%
        \includegraphics[height=\imageheight]{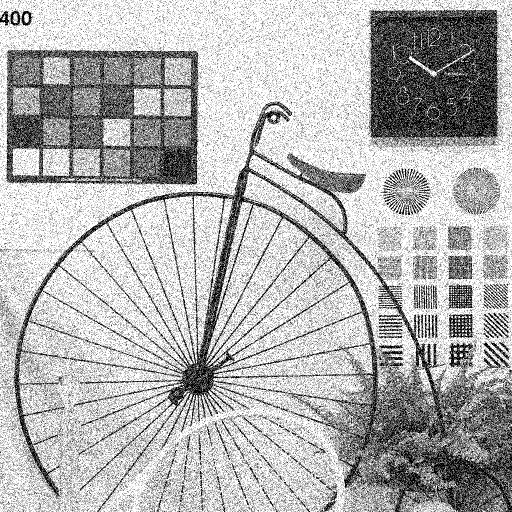}\llap{\makebox[\wd1][l]{\fbox{\includegraphics[height=\detailwidth]{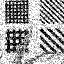}}}}%
        \subcaption*{Noisy version}%
    \end{subfigure}%
    \hfill%
    \begin{subfigure}{0.3\linewidth}%
        \includegraphics[height=\imageheight]{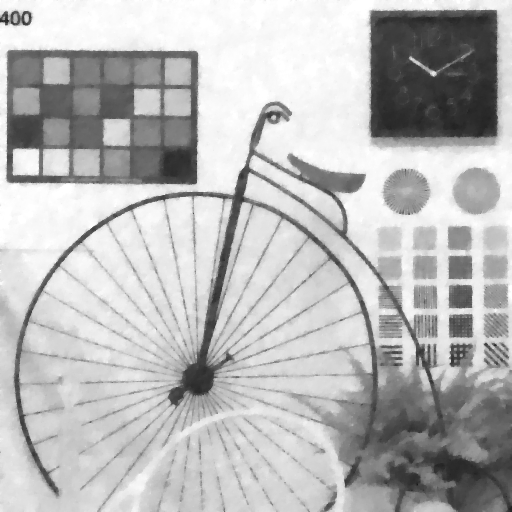}\llap{\makebox[\wd1][l]{\fbox{\includegraphics[width=\detailwidth]{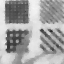}}}}%
        \subcaption*{MMSE estimate}%
    \end{subfigure}%
    \begin{subfigure}{0.08\linewidth}%
        \includegraphics[height=\imageheight]{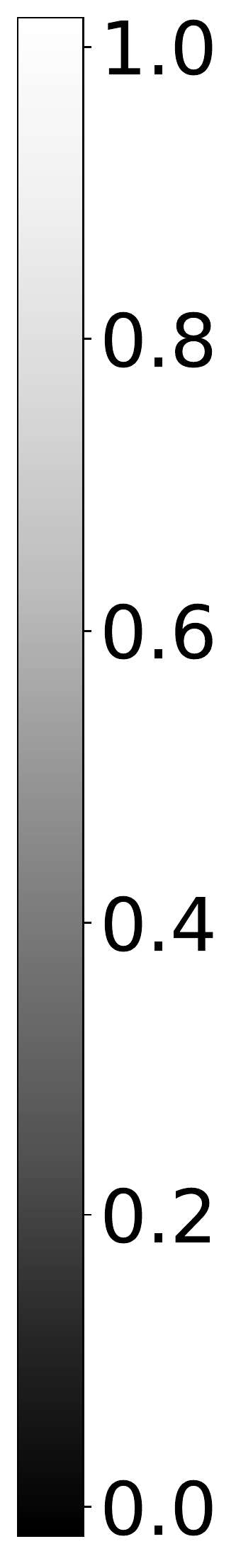}%
        \subcaption*{}
    \end{subfigure}%
    \hfill%
    \caption{Test image for the TV-denoising experiment. The MMSE is computed using $10^5$ samples \cref{algo:iPGLA} with a small error threshold ($\tilde \epsilon = 10^{-5}$).}
    \label{fig:TV_denoising_wheel_referenceimages}
\end{figure}

To illustrate the effect of the inexactness on the regularity of the samples, we report the point estimates of the posterior's first and second moment for different values of $\tilde \epsilon$ in \cref{fig:TV_denoising_results_largest_stepsize}. Note that in the particular way we implemented the inner iteration, large values of $\tilde \epsilon$ result in approximating the proximal mapping by the identity and hence in little regularizing effect. The bias introduced by the errors is then too large, the samples approximate a Gaussian distribution centered at the noisy image and hence the sample mean is close to the noisy image and the sample standard deviation almost constant. With smaller errors ($\tilde \epsilon = 10^{-2}, 10^{-4}$ here), the TV regularization effect becomes increasingly visible, while of course the computational effort is increased significantly.

\begin{figure}[t]
    \centering
    \hfill%
    \begin{subfigure}{0.49\linewidth}%
        \includegraphics[width=\linewidth]{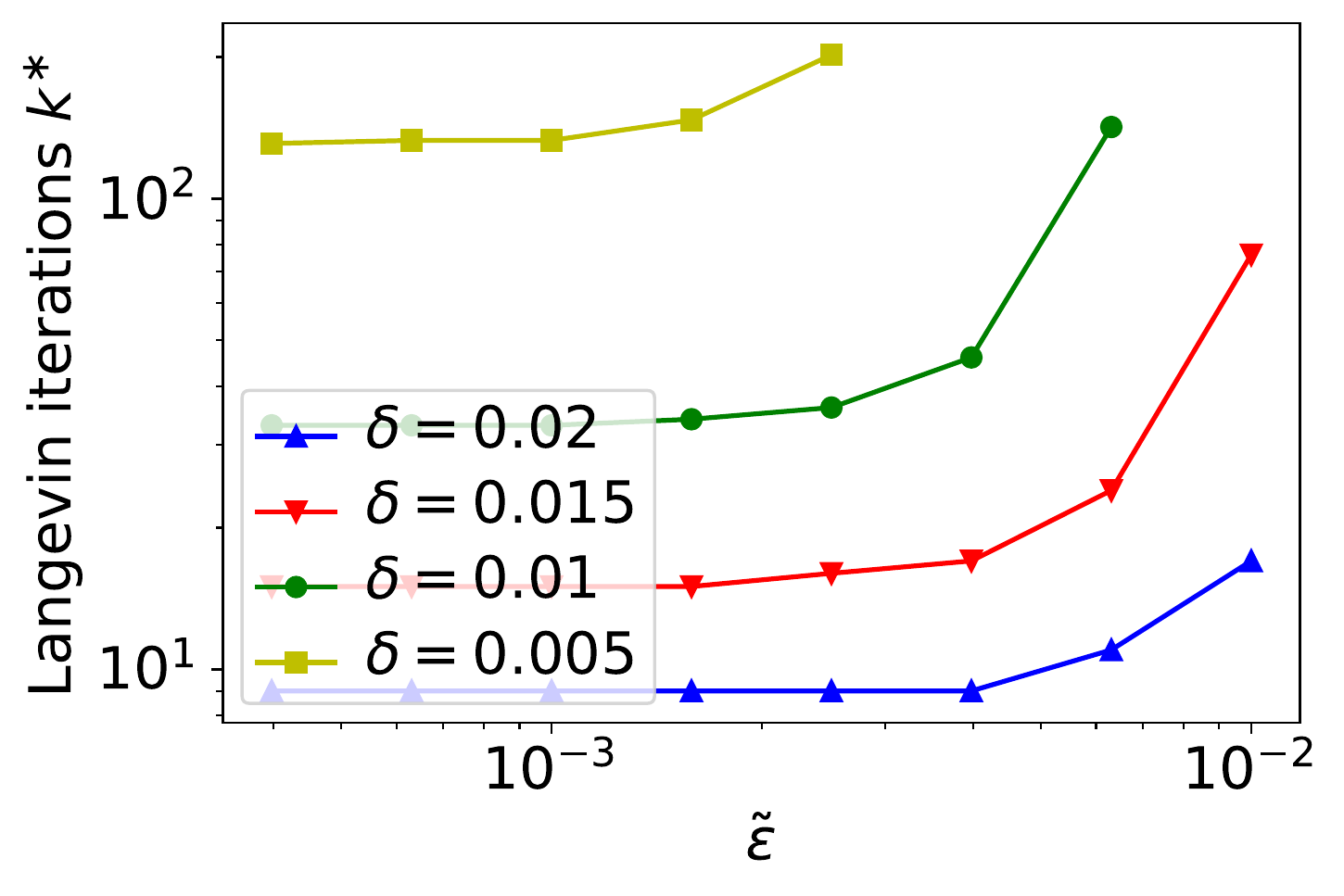}
    \end{subfigure}
    \hfill%
    \begin{subfigure}{0.49\linewidth}%
        \includegraphics[width=\linewidth]{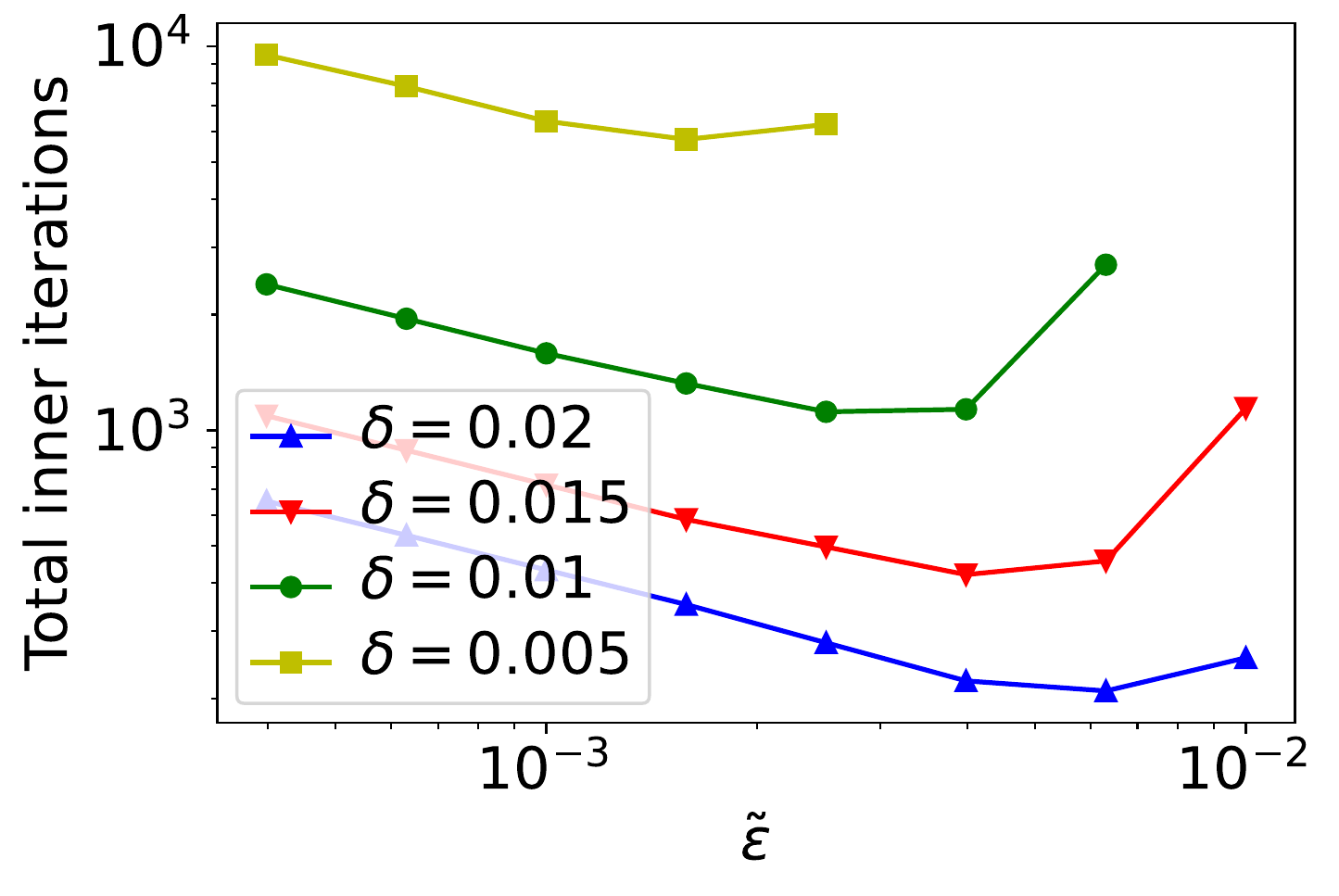}
    \end{subfigure}
    \hfill%
    \null%
    \caption{Convergence of the first moment (MMSE) at different error levels $\tilde \epsilon$. The left plot shows the number of samples necessary to reach a given accuracy $\delta$ in relative error to the MMSE estimate (no marker means that relative error $\delta$ was not reached until a maximum of $k^\ast = 10^4$ samples). The right plot shows the corresponding total number of inner iterations to compute the proximal points. It can be seen that, although less samples are necessary when $\tilde \epsilon$ is small, computationally it can be of advantage to allow larger errors in the proximal points.}\label{fig:convergence_MMSE_TV_denoising}
\end{figure}


\begin{figure}[t]
    \centering
    \setlength{\fboxsep}{0pt}%
    \setlength{\fboxrule}{1pt}%
    \setlength{\imageheight}{0.31\linewidth}%
    \setlength{\detailwidth}{0.08\linewidth}%
    \setbox1=\hbox{\includegraphics[height=\imageheight]{results/denoise_TV/wheel/mmse.png}}%
    \hfill%
    \begin{subfigure}{0.31\linewidth}%
        \includegraphics[height=\imageheight]{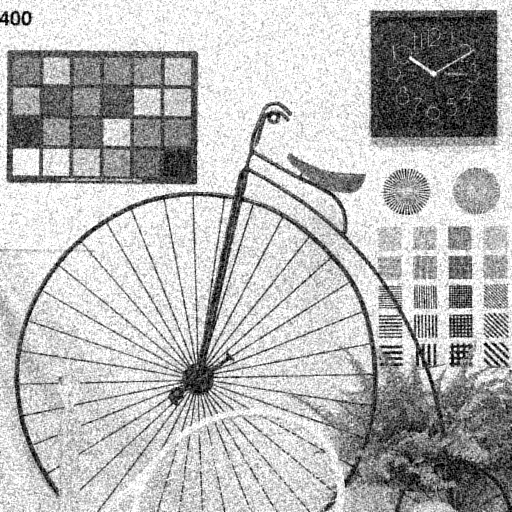}\llap{\makebox[\wd1][l]{\fbox{\includegraphics[height=\detailwidth]{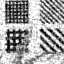}}}}%
    \end{subfigure}%
    \hfill%
    \begin{subfigure}{0.31\linewidth}%
        \includegraphics[height=\imageheight]{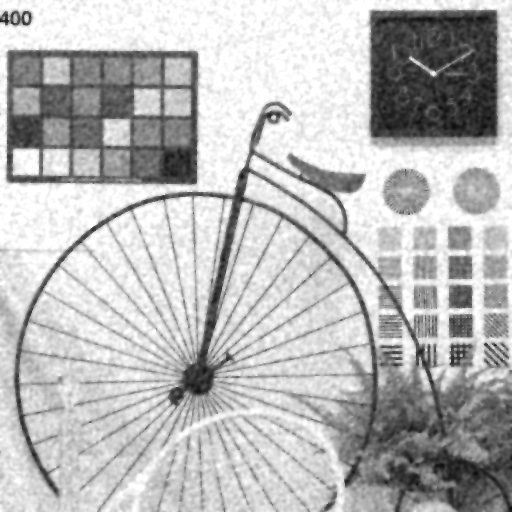}\llap{\makebox[\wd1][l]{\fbox{\includegraphics[height=\detailwidth]{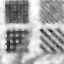}}}}%
    \end{subfigure}%
    \hfill%
    \begin{subfigure}{0.31\linewidth}%
        \includegraphics[height=\imageheight]{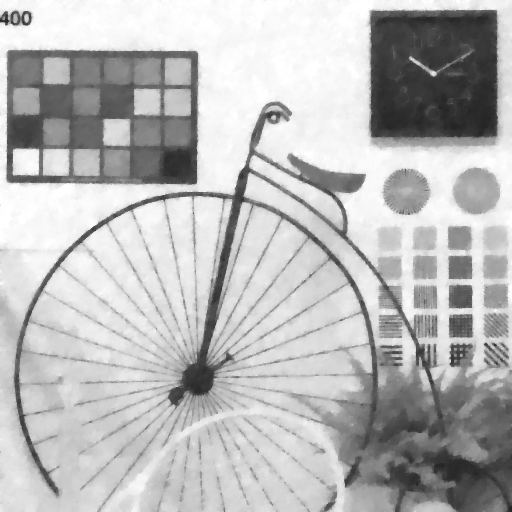}\llap{\makebox[\wd1][l]{\fbox{\includegraphics[height=\detailwidth]{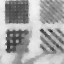}}}}%
    \end{subfigure}%
    \hfill%
    \begin{subfigure}{0.06\linewidth}%
        \includegraphics[height=\imageheight]{results/denoise_TV/wheel/cbar-mean.pdf}%
    \end{subfigure}%
    \hfill\\%
    \hfill%
    \begin{subfigure}{0.31\linewidth}%
        \includegraphics[height=\imageheight]{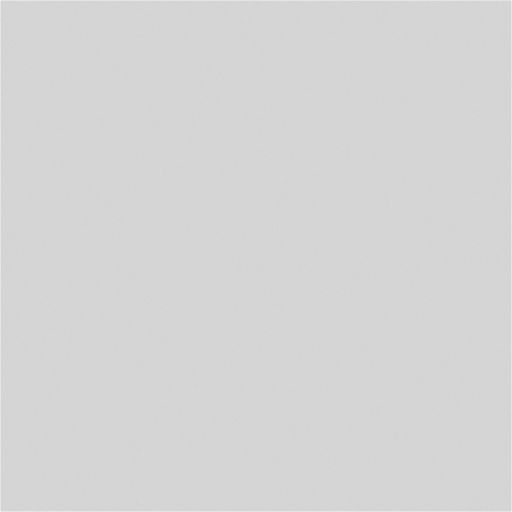}\llap{\makebox[\wd1][l]{\fbox{\includegraphics[height=\detailwidth]{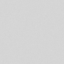}}}}%
        \subcaption*{$\tilde \epsilon = 10^{0}$}%
    \end{subfigure}%
    \hfill%
    \begin{subfigure}{0.31\linewidth}%
        \includegraphics[height=\imageheight]{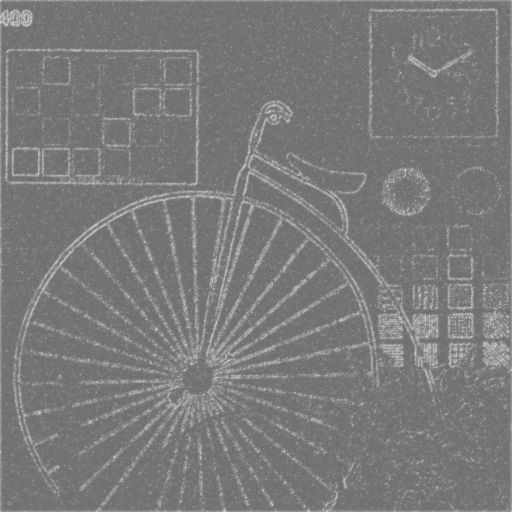}\llap{\makebox[\wd1][l]{\fbox{\includegraphics[height=\detailwidth]{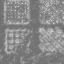}}}}%
        \subcaption*{$\tilde \epsilon = 10^{-2}$}%
    \end{subfigure}%
    \hfill%
    \begin{subfigure}{0.31\linewidth}%
        \includegraphics[height=\imageheight]{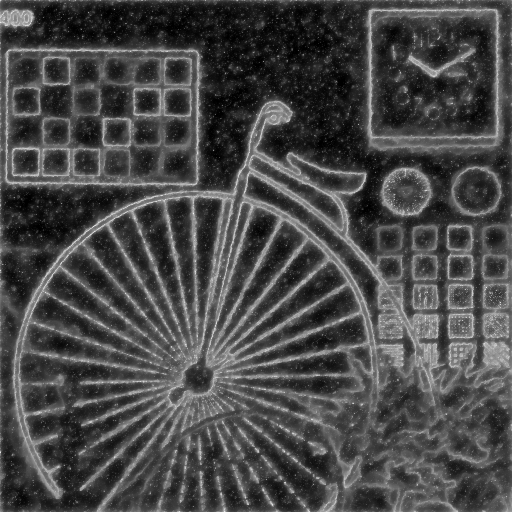}\llap{\makebox[\wd1][l]{\fbox{\includegraphics[height=\detailwidth]{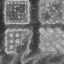}}}}%
        \subcaption*{$\tilde \epsilon = 10^{-4}$}%
    \end{subfigure}%
    \hfill%
    \begin{subfigure}{0.06\linewidth}%
       \includegraphics[width=\linewidth]{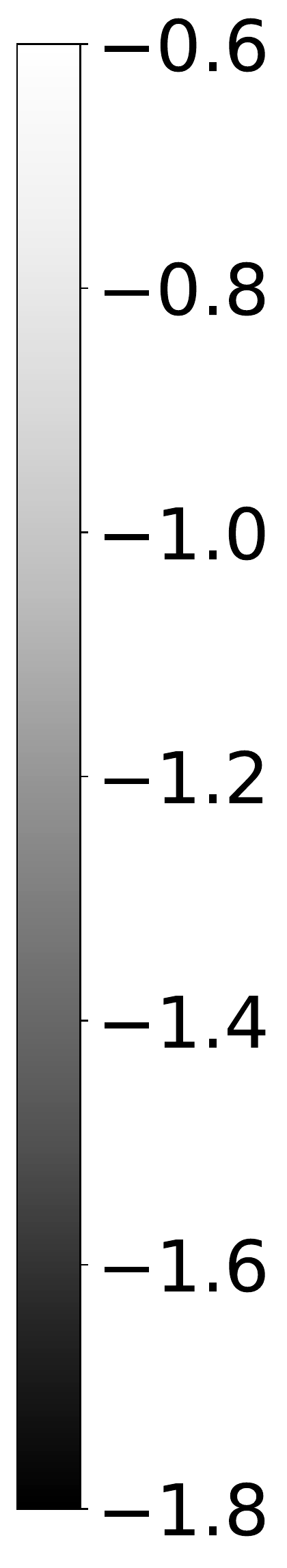}%
       \subcaption*{}
    \end{subfigure}%
    \hfill\\%
    
    \caption{TV-denoising: MMSE estimates (top) and logarithm of pixelwise standard deviation (bottom) of $10^4$ samples for different proximal accuracy levels $\epsilon = C_0 \tilde \epsilon$.}
    \label{fig:TV_denoising_results_largest_stepsize}
\end{figure}

\subsection{TV-based deblurring with highly ill-conditioned blur kernel}\label{subsec:tv_deblurring_gauss}
As a more ill-posed problem we consider image deblurring from a Gaussian blur with a TV prior. The blur kernel has a standard deviation of $1.5$ pixels and we add normally distributed noise ($\sigma = 0.1$) to the blurred image. Since the blur operator is ill-conditioned, the posterior distribution has a fairly large region of high probability compared to the largest possible stepsize $L^{-1} = \sigma^2/\lambda_{\max}(A^\ast A)$. The Markov chain therefore has a longer burn-in time and we have to run the sampling algorithm longer (compared to the denoising example) in order to get representative samples from the whole posterior distribution. We compute $10^5$ samples and reduce the computational complexity by considering a $256\times256$ image. The results including MMSE estimates and pixelwise estimated standard deviation are shown in \cref{fig:TV_deblur_owl_images}. In the deblurring tests the number of steps necessary to achieve a sufficient approximation of the proximal points was experienced to be very small, presumably because the proximal points are very close to the initial guess $X^{k+2/3}$ in each step due to the smaller step size $L^{-1}$ and the smaller regularization parameter. In particular, the results in \cref{fig:TV_deblur_owl_images} were computed with $\tilde \epsilon = 10^{-4}$ corresponding to on average 53 inner steps of gradient descent on \eqref{eq:dual_form_ROF_problem} in every sampling iteration. With larger errors $\tilde \epsilon = 10^{0}$ and $\tilde \epsilon = 10^{-2}$ (corresponding to on average 4 resp. 16 inner iterations for every sample), the considered point estimates barely changed, e.g., the PSNR of the posterior mean changed by $0.01\%$ from the experiment with $\tilde \epsilon = 10^{0}$ to the one with $\tilde \epsilon = 10^{-4}$.

\begin{figure}[t]
    \setlength{\imageheight}{0.31\linewidth}%
    \centering
    \hfill%
    \begin{subfigure}{0.31\linewidth}%
        \centering%
        \includegraphics[height=\imageheight]{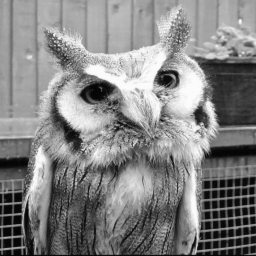}%
    \end{subfigure}%
    \hfill%
    \begin{subfigure}{0.31\linewidth}%
        \centering%
        \includegraphics[height=\imageheight]{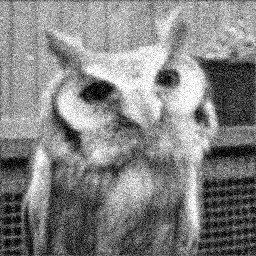}%
    \end{subfigure}%
    \hfill%
    \begin{subfigure}{0.31\linewidth}%
        \centering%
        \includegraphics[height=\imageheight]{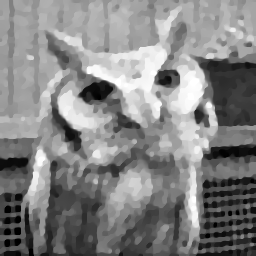}%
    \end{subfigure}%
    \begin{subfigure}{0.06\linewidth}%
       \includegraphics[height=\imageheight]{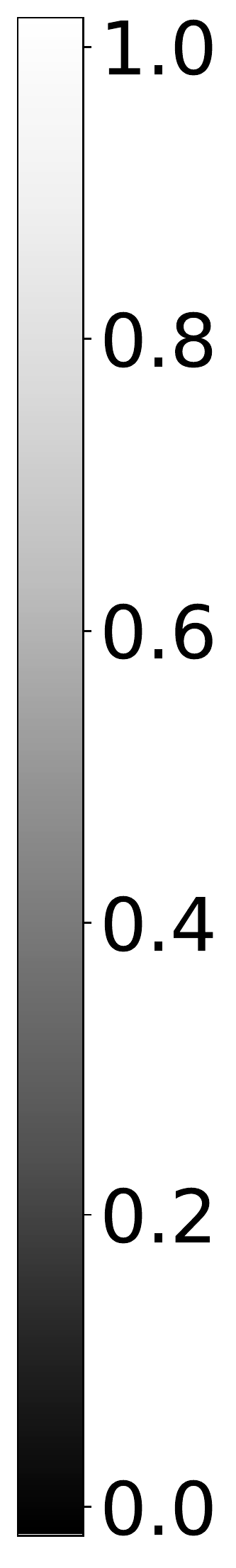}%
    \end{subfigure}%
    \hfill%
    \null%
    \\%
    \vspace{1em}
    \hfill%
    \begin{subfigure}{0.31\linewidth}%
        \centering%
        \includegraphics[height=\imageheight]{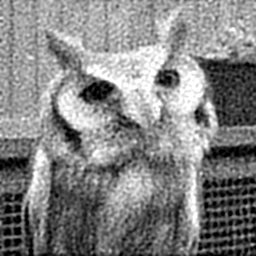}
    \end{subfigure}%
    \begin{subfigure}{0.06\linewidth}%
       \includegraphics[height=\imageheight]{results/deblur_TV/owl/cbar-mean.pdf}%
    \end{subfigure}%
    \begin{subfigure}{0.31\linewidth}%
        \centering%
        \includegraphics[height=\imageheight]{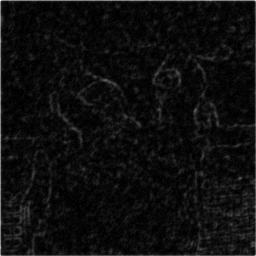}%
    \end{subfigure}%
    \begin{subfigure}{0.08\linewidth}%
       \includegraphics[height=\imageheight]{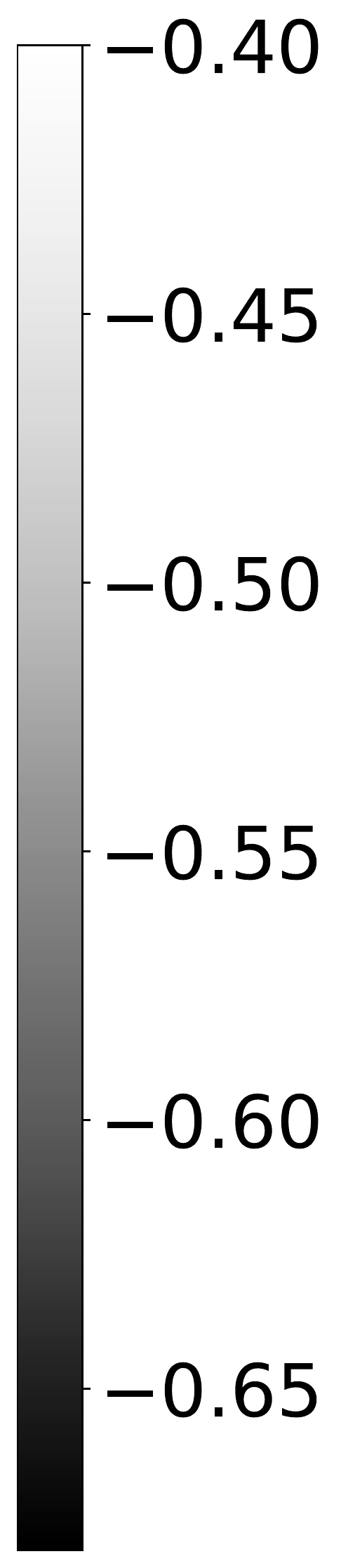}%
    \end{subfigure}%
    \hfill%
    \null%
    \caption{TV-deblurring experiment. First row from left to right: ground truth, blurred and noisy observation and MAP estimate. Second row: MMSE estimate computed as the mean of samples and (log-scaled) pixel-wise standard deviation generated by \cref{algo:iPGLA} with $\log_{10} \tilde \epsilon = -4$.}
    \label{fig:TV_deblur_owl_images}
\end{figure}

\subsection{TV-based deblurring from low-count Poisson data}\label{subsec:tv_deblurring_poisson}
Finally, we illustrate the applicability of the method to posteriors arising from different and more challenging log-likelihood functions, in particular such with an admissible set of the samples. We consider the problem of deblurring from Poisson-distributed data. For this purpose, we assume that observed data $y$ is distributed according to the model
$$ y|x \sim \mathrm{Poisson}(Ax+\sigma), $$
with a forward model $A \in \bbR^{m\times d}$ and we try to recover the posterior distribution of $x$. In photon counting applications, the parameter $\sigma \in \bbR^m_+$ is typically interpreted as a background level and is assumed to be known here. Computing the log-likelihood results (up to additive constants) in the potential term
$$ F(x) = \sum_{i=1}^m (Ax)_i - y_i \log((Ax)_i + \sigma_i). $$
The image $x$ is interpreted as some (nonnegative) intensity, i.e. we are only interested in $x\ge 0$. Since $A$ is positivity-preserving (i.e. $x\ge 0$ implies $Ax \ge 0$), $F$ is well-defined on the admissible set $\bbR^d_+$. In our algorithmic approach, we include the constraint directly in the nonsmooth potential term $G$ by adding an indicator to the TV regularization term:
$$ G(x) = \mu_{\TV}\TV(x) + \iota_{\bbR^d_+}(x). $$
Formally, $G(x) = \infty$ for $x \notin \bbR_+^d$, which is identified with $\mu^\ast(x) = 0$. Note that there exist other approaches to ensuring the nonnegativity of samples generated by Langevin algorithms, e.g.\ in \cite{Melidonis2023} the authors modify the original Langevin SDE \eqref{eq:LangevinSDE} by incorporating an additional reflection map. Technically, since $G$ is not differentiable, the diffusion process we consider here is defined by a stochastic differential inclusion variant of \eqref{eq:LangevinSDE} since $\nabla V$ needs to be replaced by $\nabla F + \partial G$. With our choice of potentials, \cref{assumption1,assumption2,assumption3} are satisfied if $\sigma > 0$, the iteration is well-defined and the samples $X^k$ drawn by \cref{algo:iPGLA} are always nonnegative, which allows improved interpretability. 

We choose $A$ to be a uniform blur operator and assume fixed background level $\sigma_i = 0.02$. The $256\times256$ test image is scaled to a mean intensity value of 2 across all pixels, so that the observed data $y$ has a low count rate and hence moderately small signal to noise ratio, making the problem ill-conditioned. The step size choice is more challenging in the Poisson likelihood case. Since the global Lipschitz constant $L$ of $\nabla F$ becomes large for small $\sigma$, choosing the step size as $1/L$ results in small step sizes that are infeasible in practice. In order to compare influence of step size and proximal errors on the results, we also run Markov chains using backtracking line search as is common in first order optimization schemes \cite{Armijo1966} and pick $\gamma_k$ as large as possible in a geometric range of values while still satisfying \eqref{eq:descent_condition} at every iteration (two-way backtracking). The inexactness levels of $\prox_{\gamma_k G}$ are chosen implicitly by running fixed numbers of inner iterations (note that \cref{lem:guarantee_epsilon_approximation_by_duality_gap} can not be used since $G$ does not have the necessary form). As before, we report estimates of the posterior mean and (log-scaled) standard deviation. We also quantify the uncertainty in the posterior at different spatial scales, by downsampling every sample image by factors 2, 4 and 8 and computing the resulting standard deviation maps. In order to quantify the mixing behaviour of the Markov chains, we further plot the sample autocorrelation functions (ACF) of the fastest and slowest mixing components. Since the computation of the posterior covariance matrix is infeasible in large dimensions, it is approximated by assuming the posterior covariance has the same eigenvectors as the blur operator (i.e., that its eigenvectors are the Fourier modes). 

We simulate a blurred and noisy observation and draw $10^6$ samples from the posterior using \cref{algo:iPGLA} in three different configurations. The first uses fixed step sizes $\gamma = 1/\tilde L$, where $\tilde L$ is an estimate of the global Lipschitz constant of $\nabla F$. The second and the third use backtracking line search and choose considerably larger step sizes ($> 10^3$ times larger on average). We account for the large difference in step sizes only by increasing the burn-in period of the first chain. The second and third experiment are different in the proximal inexactness level, one using 10 inner iterations of AGD for computing $\prox_{\gamma_k G}$, and one only a single inner iteration. The resulting mean and standard deviation estimates in \cref{fig:deblur_poisson_TV} and \cref{fig:deblur_poisson_TV_second moments} confirm the common effect that in view of numerical efficiency, picking a larger step size can be beneficial in order to guarantee fast convergence of point estimates.  Although in theory, small step sizes guarantee a smaller asymptotical bias of the targeted distribution, the samples are strongly correlated and it would be computationally infeasible to run the chain for long enough to mitigate this effect. The autocorrelation functions in \cref{fig:acfs_poisson_deblurring} confirm this behaviour. Compared to these errors, with large step sizes the additional effect of computing the proximal mapping only with low accuracy has a minor effect on the convergence behaviour of the chain and resulting point estimates. The MMSE estimate PSNR is 0.1 worse on the experiment with a single inner iteration compared to the 10 inner iterations configuration, and only the slowest components of the Markov chain show a slightly stronger correlation along the iterates. The standard deviation maps on pixel scale and larger spatial scales are hardly influenced by larger proximal errors either. 

\setlength{\imageheight}{0.31\linewidth}%
\begin{figure}[t]
    \centering%
    \hfill%
    \begin{subfigure}{0.31\linewidth}%
        \includegraphics[width=\linewidth]{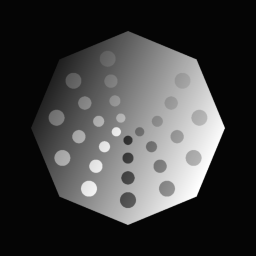}
        \subcaption*{Ground truth}%
    \end{subfigure}%
    \hfill%
    \begin{subfigure}{0.31\linewidth}%
        \includegraphics[width=\linewidth]{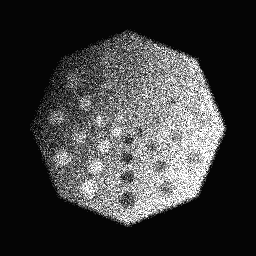}
        \subcaption*{Blurred and noisy}%
    \end{subfigure}%
    \hfill%
    \begin{subfigure}{0.06\linewidth}%
        \includegraphics[height=\imageheight]{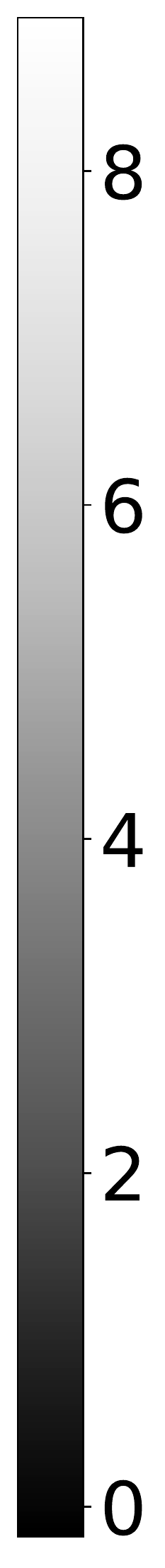}%
        \subcaption*{}
    \end{subfigure}%
    \\%
    \hfill%
    \begin{subfigure}{0.31\linewidth}%
        \includegraphics[width=\linewidth]{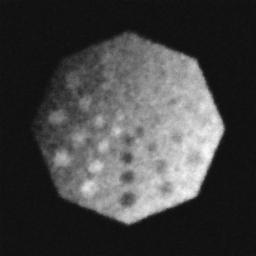}
        \subcaption*{Fixed $\gamma$, 10 inner iterations\\%
        MMSE estimate: PSNR=23.45}%
    \end{subfigure}%
    \hfill%
    \begin{subfigure}{0.31\linewidth}%
        \includegraphics[width=\linewidth]{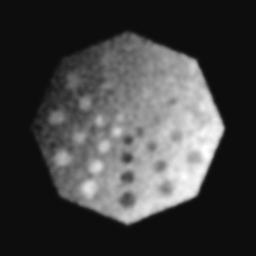}
        \subcaption*{Backtracking, 1 inner it.\\%
        MMSE estimate: PSNR=25.55}%
    \end{subfigure}%
    \hfill%
    \begin{subfigure}{0.31\linewidth}%
        \includegraphics[width=\linewidth]{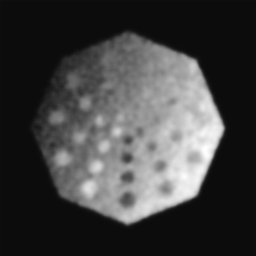}%
        \subcaption*{Backtracking, 10 inner its.\\%
        MMSE estimate: PSNR=25.65}%
    \end{subfigure}%
    \hfill%
    \begin{subfigure}{0.06\linewidth}%
        \includegraphics[height=\imageheight]{results/deblur_poisson/phantom/phantom256_colorbar.pdf}%
        \subcaption*{\ \\ \ }
    \end{subfigure}%
    \hfill%
    \null%
    \caption{MMSE estimates produced by three different configurations of \cref{algo:iPGLA} on the Poisson data deblurring experiment.}
    \label{fig:deblur_poisson_TV}
\end{figure}

\begin{figure}[t]
    \centering
    \setlength{\imageheight}{0.2\linewidth}
    \hfill%
    \begin{subfigure}{0.20\linewidth}%
        \includegraphics[height=\imageheight]{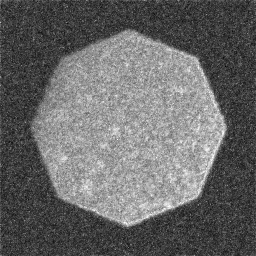}
    \end{subfigure}%
    \hfill%
    \begin{subfigure}{0.04\linewidth}%
        \includegraphics[height=\imageheight]{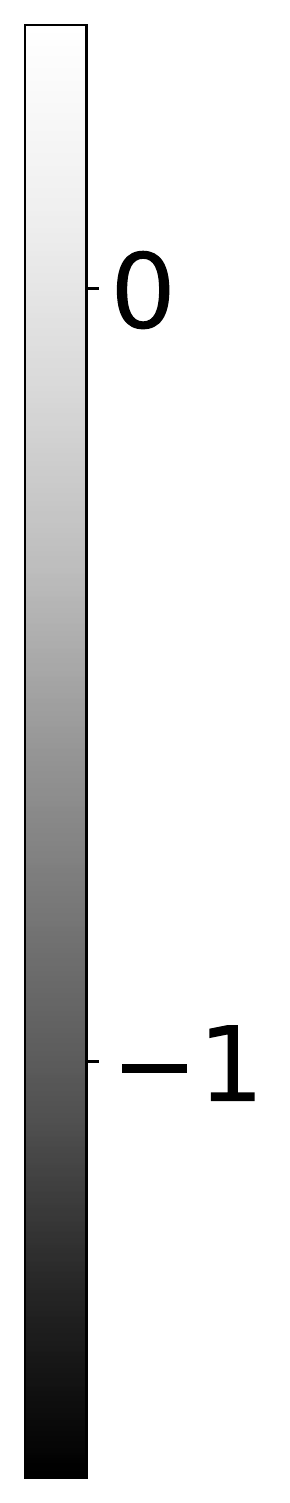}
    \end{subfigure}%
    \hfill%
    \begin{subfigure}{0.20\linewidth}%
        \includegraphics[height=\imageheight]{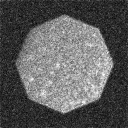}
    \end{subfigure}%
    \hfill%
    \begin{subfigure}{0.04\linewidth}%
        \includegraphics[height=\imageheight]{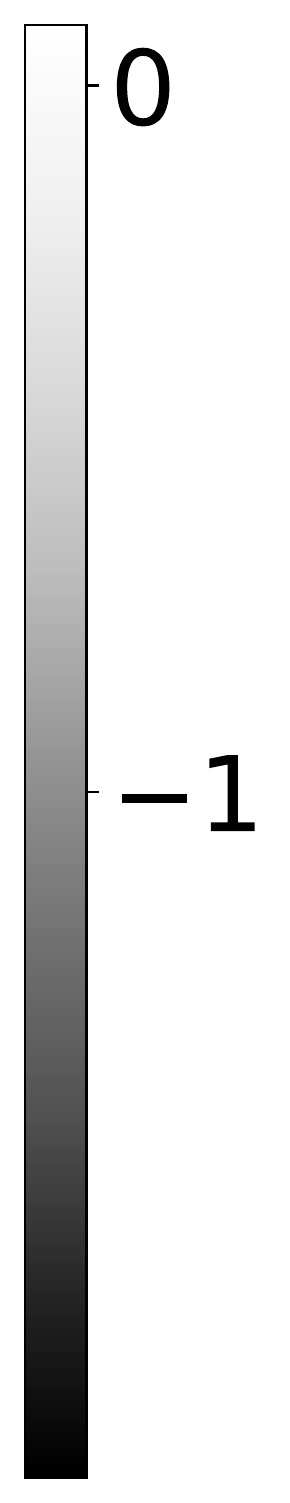}
    \end{subfigure}%
    \hfill%
    \begin{subfigure}{0.20\linewidth}%
        \includegraphics[height=\imageheight]{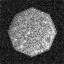}
    \end{subfigure}%
    \hfill%
    \begin{subfigure}{0.04\linewidth}%
        \includegraphics[height=\imageheight]{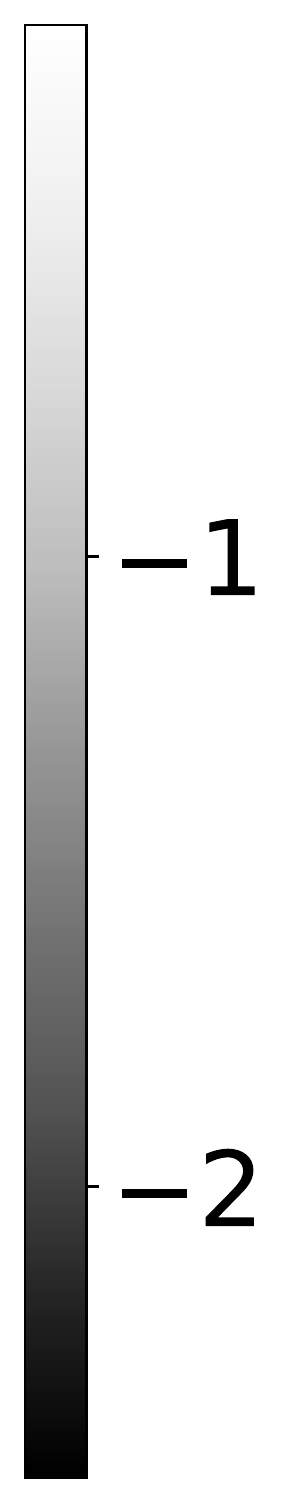}
    \end{subfigure}%
    \hfill%
    \begin{subfigure}{0.20\linewidth}%
        \includegraphics[height=\imageheight]{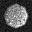}
    \end{subfigure}%
    \hfill%
    \begin{subfigure}{0.04\linewidth}%
        \includegraphics[height=\imageheight]{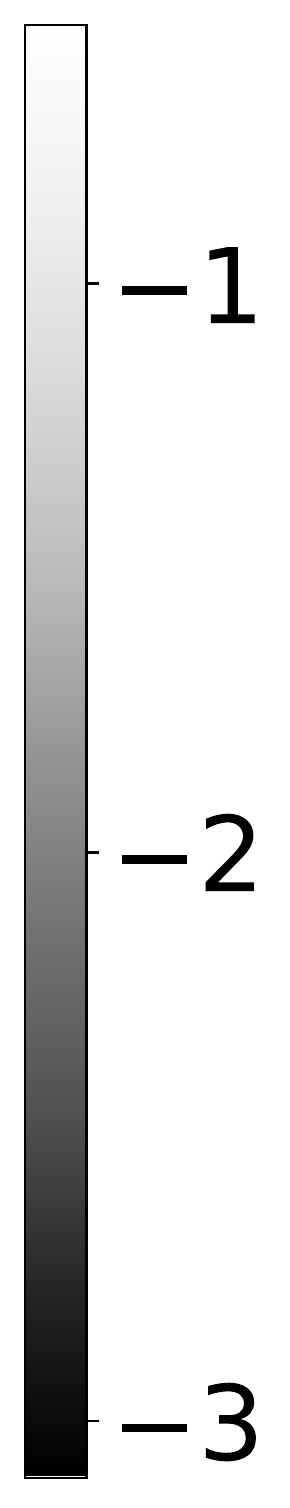}
    \end{subfigure}%
    \hfill%
    \\%
    \hfill%
    \begin{subfigure}{0.20\linewidth}%
        \includegraphics[height=\imageheight]{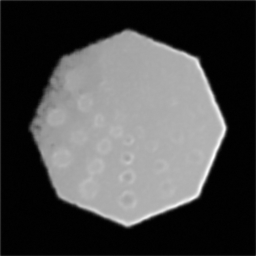}
    \end{subfigure}%
    \hfill%
    \begin{subfigure}{0.04\linewidth}%
        \includegraphics[height=\imageheight]{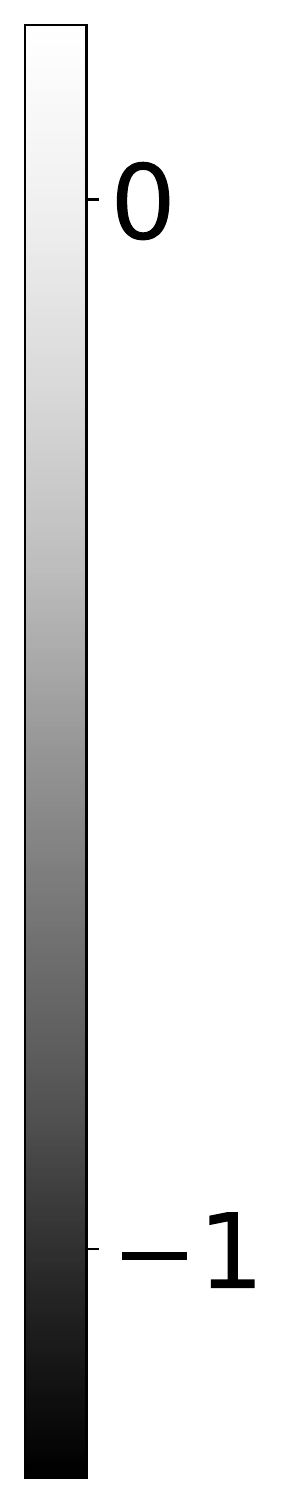}
    \end{subfigure}%
    \hfill%
    \begin{subfigure}{0.20\linewidth}%
        \includegraphics[height=\imageheight]{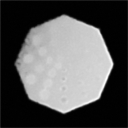}
    \end{subfigure}%
    \hfill%
    \begin{subfigure}{0.04\linewidth}%
        \includegraphics[height=\imageheight]{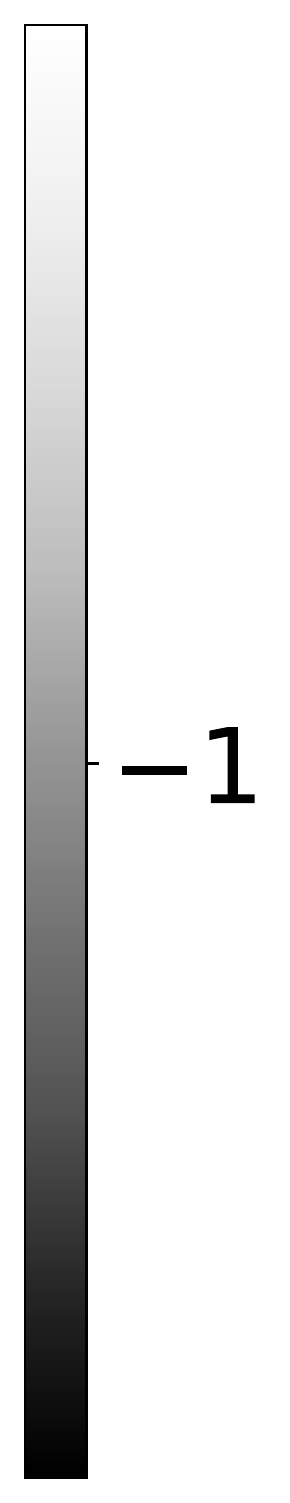}
    \end{subfigure}%
    \hfill%
    \begin{subfigure}{0.20\linewidth}%
        \includegraphics[height=\imageheight]{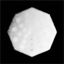}
    \end{subfigure}%
    \hfill%
    \begin{subfigure}{0.04\linewidth}%
        \includegraphics[height=\imageheight]{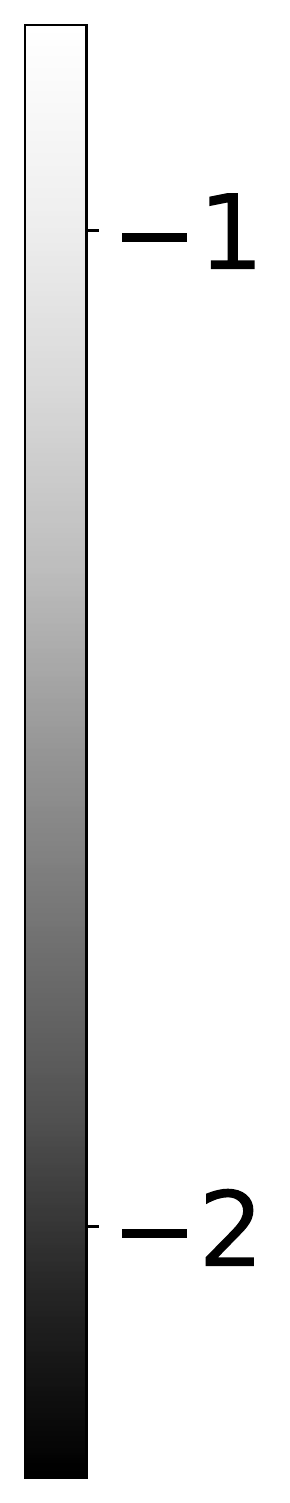}
    \end{subfigure}%
    \hfill%
    \begin{subfigure}{0.20\linewidth}%
        \includegraphics[height=\imageheight]{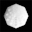}
    \end{subfigure}%
    \hfill%
    \begin{subfigure}{0.04\linewidth}%
        \includegraphics[height=\imageheight]{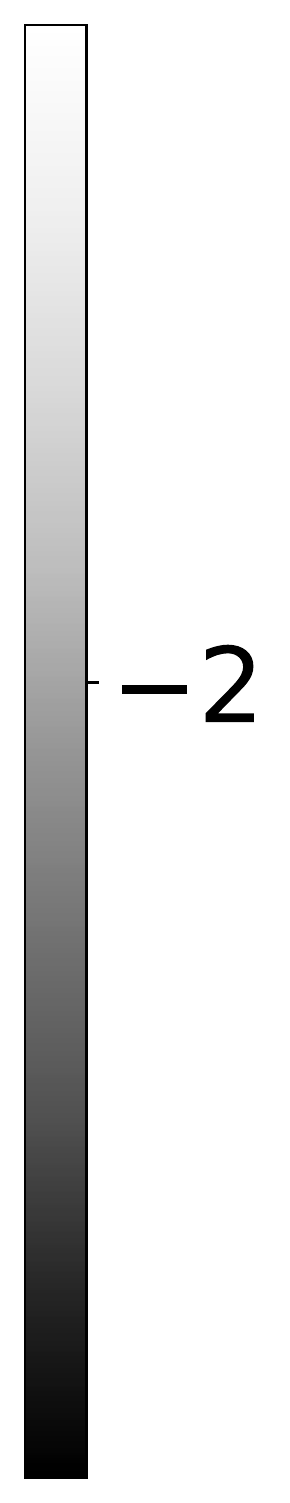}
    \end{subfigure}%
    \hfill%
    \\%
    \hfill%
    \begin{subfigure}{0.20\linewidth}%
        \includegraphics[height=\imageheight]{results/deblur_poisson/phantom/logstd_scale0_phantom256_2e+00miv_1proxiters_1e+00regparam_btsteps_1e+06samples.png}
    \end{subfigure}%
    \hfill%
    \begin{subfigure}{0.04\linewidth}%
        \includegraphics[height=\imageheight]{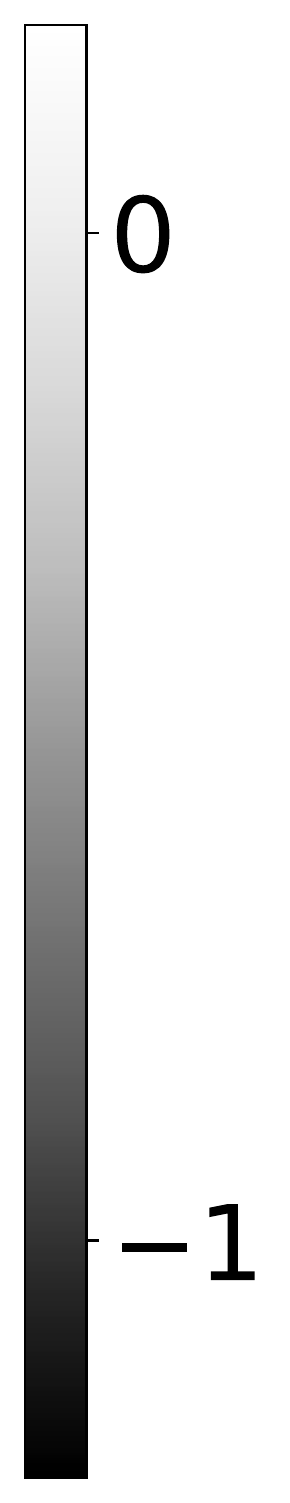}
    \end{subfigure}%
    \hfill%
    \begin{subfigure}{0.20\linewidth}%
        \includegraphics[height=\imageheight]{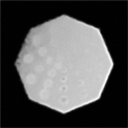}
    \end{subfigure}%
    \hfill%
    \begin{subfigure}{0.04\linewidth}%
        \includegraphics[height=\imageheight]{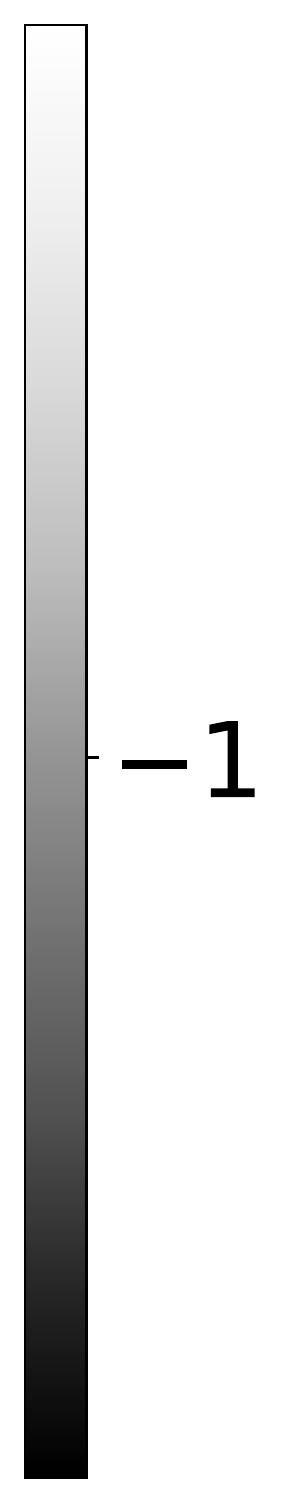}
    \end{subfigure}%
    \hfill%
    \begin{subfigure}{0.20\linewidth}%
        \includegraphics[height=\imageheight]{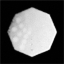}
    \end{subfigure}%
    \hfill%
    \begin{subfigure}{0.04\linewidth}%
        \includegraphics[height=\imageheight]{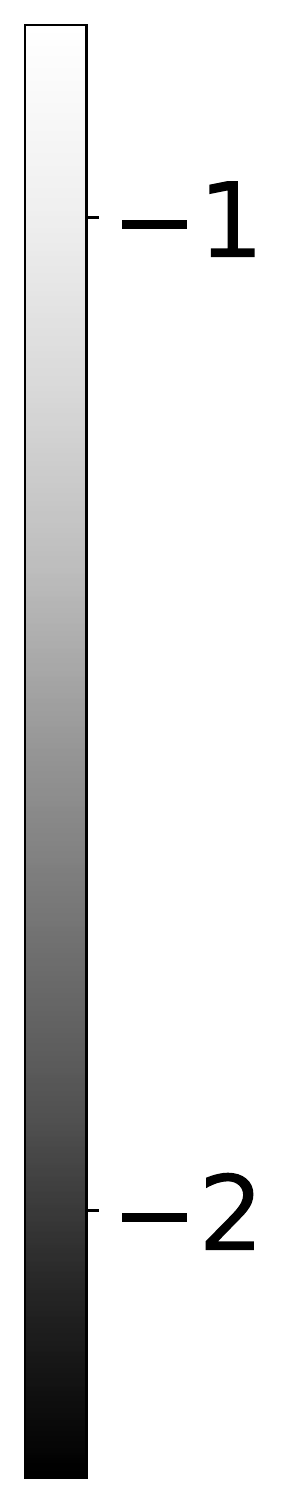}
    \end{subfigure}%
    \hfill%
    \begin{subfigure}{0.20\linewidth}%
        \includegraphics[height=\imageheight]{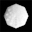}
    \end{subfigure}%
    \hfill%
    \begin{subfigure}{0.04\linewidth}%
        \includegraphics[height=\imageheight]{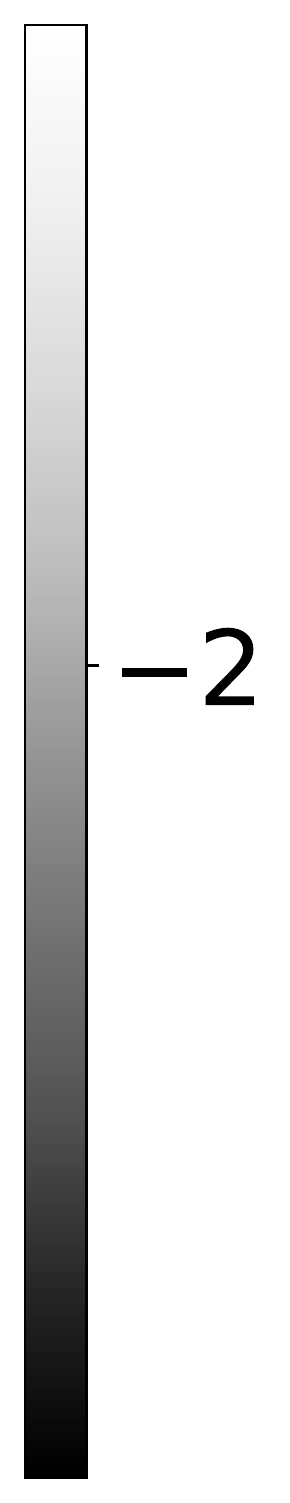}
    \end{subfigure}%
    \hfill%
    \caption{Reconstructed second moment of Poisson deblurring posterior. From left to right, we display pixel-wise standard deviation (logarithmic color scale) and standard deviation of the sample images downsampled by factors 2, 4 and 8. Top row is computed with fixed step $1/\tilde L$ and 10 inner iterations, middle row with backtracking step sizes and 1 inner iteration, bottom row with backtracking step sizes and 10 inner iterations.}
    \label{fig:deblur_poisson_TV_second moments}
\end{figure}

\begin{figure}
    \centering
    \includegraphics[width=\linewidth]{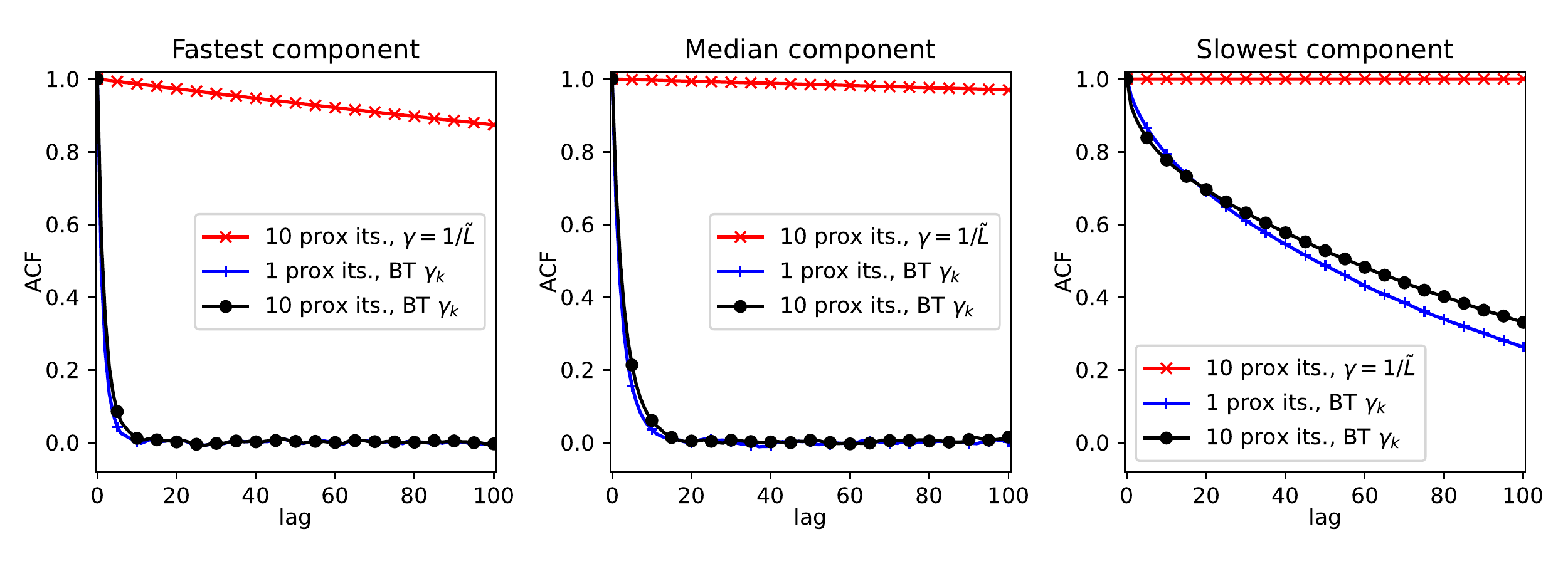}
    \caption{Autocorrelation functions of the slowest, median and fastest component of the Markov chain in the deblurring from Poisson data experiment described in \cref{subsec:tv_deblurring_poisson}.}
    \label{fig:acfs_poisson_deblurring}
\end{figure}
\section{Conclusion}
In this work we proposed a new framework of proximal Langevin sampling for inexact proximal operators. We generalized existing non-asymptotic and asymptotic convergence results on the exact proximal gradient Langevin algorithm to the inexact case. The additional bias between the stationary distribution and the target due to errors is quantified for types of errors that can be ensured efficiently in numerical settings. Our numerical experiments reflect the theoretical results and show how the error in the proximal points can be traded for speed of the algorithm. In particular, in regimes where the step size dominates the algorithm's bias, it can be useful to stop the inner iteration approximating the proximal map after very few iterations giving efficient sampling schemes in high-dimensional problems.

\section*{Acknowledgements}
MJE acknowledges support from EPSRC (EP/S026045/1, EP/T026693/1, EP/V026259/1) and the Leverhulme Trust (ECF-2019-478). LK acknowledges support from the German Federal Ministry of Education and Research BMBF (15S59431 B). LK and CBS acknowledge support from the European Union Horizon 2020 research and innovation programme under the Marie Sk{\l}odowska-Curie grant agreement No. 777826 NoMADS. CBS further acknowledges support from the Philip Leverhulme Prize, the Royal Society Wolfson Fellowship, the EPSRC advanced career fellowship EP/V029428/1, EPSRC grants EP/S026045/1 and EP/T003553/1, EP/N014588/1, EP/T017961/1, the Wellcome Innovator Awards 215733/Z/19/Z and 221633/Z/20/Z, the Cantab Capital Institute for the Mathematics of Information and the Alan Turing Institute.

\printbibliography
\end{document}